\documentclass[10pt, twoside, twocolumn]{IEEEtran}

\usepackage[latin1]{inputenc}
\usepackage[T1]{fontenc}
\usepackage{amssymb,mathtools}   
\usepackage{amsthm}
\usepackage{graphicx}
\usepackage{psfrag}
\usepackage{color}

\usepackage{adjustbox}
\usepackage{suffix}
\usepackage{booktabs}  
\usepackage{multirow} 
\usepackage{relsize}
\newtheorem{prop}{Proposition}
\newtheorem{remark}{Remark}

\newcommand*{\field}[1]{\mathbb{#1}}%

\DeclarePairedDelimiter\abs{\lvert}{\rvert} 
\usepackage{float}
\usepackage{textcomp}

\begin{document}



\title{\LARGE \bf On the Statistics of the Ratio of Non-Constrained Arbitrary\\ $\alpha$-$\mu$ Random Variables: a General Framework and Applications}


%

\author{J.~D.~Vega~S\'anchez, D.~P.~Moya~Osorio, E.~E.~Benitez~Olivo, H. Alves, M.C.P. Paredes, and L.~Urquiza-Aguiar
\thanks{J.~D.~Vega~S\'anchez, M.C.P. Paredes, and L.~Urquiza-Aguiar are with the 
Departamento de Electr\'onica, Telecomunicaciones y Redes de Informaci\'on, Escuela Polit\'ecnica Nacional (EPN),
Quito,  170525, Ecuador. (e-mail: jose.vega01@epn.edu.ec; cecilia.paredes@epn.edu.ec; luis.urquiza@epn.edu.ec)}
\thanks{D.~P.~Moya~Osorio is with the Department of Electrical Engineering, Center of Exact Sciences and Technology, Federal University of S\~{a}o Carlos, S\~{a}o Carlos, SP, 13565-905, Brazil. (e-mail: dianamoya@ufscar.br)}
\thanks{ E.~E.~Benitez~Olivo is with S\~{a}o Paulo State University (UNESP), Campus of S\~{a}o Jo\~{a}o da Boa Vista, 13876-750 S\~{a}o Jo\~{a}o da Boa Vista, Brazil (e-mail: edgar.olivo@unesp.br).}
\thanks{H.~Alves is with the Centre for Wireless Communications (CWC), University of Oulu, Finland. (e-mail: Hirley.Alves@oulu.fi).}
}

\maketitle

\begin{abstract}
In this paper, we derive closed-form exact expressions for the main statistics of the ratio of squared $\alpha$$-$$\mu$ random variables, which are of interest in many scenarios for future wireless networks where generalized distributions are more suitable to fit with field data. Importantly, different from previous proposals, our expressions are general in the sense that are valid for non-constrained arbitrary values of the parameters of the $\alpha$$-$$\mu$ distribution.  Thus, the probability density function, cumulative distribution function, moment generating function, and higher order moments are given in terms of both $(i)$ the Fox H-function for which we provide a portable and efficient Wolfram Mathematica code and $(ii)$ easily computable series expansions. Our expressions can be used straightforwardly in the performance analysis of a number of wireless communication systems, including either interference-limited scenarios, spectrum sharing, full-duplex or physical-layer security networks, for which we present the application of the proposed framework. Moreover, closed-form expressions for some classical distributions, derived as special cases from the $\alpha$$-$$\mu$ distribution, are provided as byproducts. The validity of the proposed expressions is confirmed via Monte Carlo simulations.

\end{abstract}


\section{INTRODUCTION}

\IEEEPARstart{T}{he} accurate wireless channel characterization is, currently, of paramount importance for the understanding, evaluation and design of future wireless communication systems, which will operate under stringent requirements of capacity, reliability, latency, and scalability, to enable the new applications for machine-type communications (MTC) and  mission-critical communications envisioned for the scenarios of 5G and beyond. Accordingly, wireless propagations models are crucial to compare potential candidate technologies that will be used for the deployment of these networks~\cite{sun2018}. Particularly, classical small-scale fading models have shown to be limited to adequately fit experimental data from practical scenarios. Then,  more general fading models, which better capture the wireless channel statistics, are specially useful for modeling future wireless systems.

Precisely, the $\alpha$-$\mu$ fading distribution, first proposed in~\cite{yacoub2007alpha}, is a more general and flexible model to better fit with field data in cases where other widely-used classical distributions are not accurate. Also, it presents an easy mathematical tractability and includes other important distributions as special cases, such as Gamma, Nakagami-$m$, Exponential, Weibull, one-sided Gaussian, and Rayleigh. The $\alpha$-$\mu$ fading model considers a signal composed of clusters of multipath waves, which propagates in a non-linear environment. Thus, this fading model is described by two physical parameters, namely: $\alpha$ which represents the nonlinearity of the propagation environment, and $\mu$ which represents the number of multipath clusters.

The knowledge of the statistics of the sum, product, and ratio of fading random variables~(RVs) has a pivotal role in the performance analysis and evaluation of many practical wireless applications. In this context, the distribution of both the sum and the product of $\alpha$-$\mu$ RVs has been extensively studied by many research works, among them, the following are notable:~\cite{DaCosta,DaCosta2,DaCosta3,2018Naka} (for the sum), and~\cite{Carlos,product2,product3,product4} (for the product).
On the other hand, the statistics of the ratio of $\alpha$-$\mu$ RVs has been little explored in the literature, as will be shown later. It is noteworthy that, the performance analysis of wireless communications systems, specifically for scenarios considering some of the key technologies for future wireless networks, commonly involves the calculation of the ratio of signals powers, such as the signal-to-interference ratio (SIR), for instance. Therefore, the distribution of the ratio of RVs is of particular interest, and it has a pivotal role in the analytical performance evaluation of those scenarios.

Different approaches concerning the statistics of the ratio between RVs with well-known distributions such as Gamma, Exponential, Weibull, and Normal, are presented in~\cite{Ahsen,Annavajjala,Nadarajah,Gia}, where some application uses are also provided. Moreover, regarding generalized distributions, the statistics of the ratio of independent and arbitrary $\alpha$-$\mu$ RVs, via series representation, was proposed in~\cite{Leonardo}. However, in that work, the convergence of the power series was attained by making an strong assumption, more specifically: the values related to the non-linearity of the environment (i.e., to the $\alpha$ parameter, also referred as shape parameter) of each $\alpha$-$\mu$ RV involved in the quotient must be co-prime integers. Further, under the same consideration, the work in~\cite{Leonardo2016} provides closed-form expressions for the statistics of the ratio of products of an arbitrary number of independent and non-identically distributed $\alpha$-$\mu$ variates. Thus, an important constraint on the results of~\cite{Leonardo} and~\cite{Leonardo2016} is that the shape parameter (or, equivalently, the $\alpha$ parameter) of the $\alpha$-$\mu$ RVs involved on the ratio cannot take non-constrained arbitrary values. This fact hinders a more comprehensive insight into the performance analysis of different wireless communication systems. 

In light of the above considerations, in this paper we derive closed-form expressions for the main statistics of the ratio of independent and non-identically distributed (i.n.i.d.) squared $\alpha$-$\mu$ RVs, for which all the fading parameters of both distributions can be non-constrained arbitrary positive real numbers (thus including the special case of positive integers). This way, our expressions relieve the strong assumption considered in~\cite{Leonardo} and~\cite{Leonardo2016}. Also, our results can be employed as a powerful tool for the performance evaluation of different scenarios. The following are our main contributions:
\begin{itemize}
\item Novel closed-form expressions for the probability density function (PDF), cumulative distribution function (CDF), moment generating function (MGF), and higher order moments of the ratio of general $\alpha$-$\mu$ RVs are derived in terms of the Fox H-function. 
\item The statistics of the ratio of RVs for some special cases of classical fading distributions are also provided as byproducts.
\item Application uses in wireless networks are proposed in the context of Physical Layer Security (PLS), Cognitive Radio (CR), and Full-Duplex (FD) relaying, where the obtained analytical expressions can be used straightforwardly.
\item A simple, efficient and useful algorithm for the implementation of the univariate Fox H-function is also provided.
\end{itemize}


The remainder of this paper is organized as follows. Section~II revisits preliminaries on the $\alpha$-$\mu$ distribution. In Section~III, the statistics of the ratio of non-constrained arbitrary $\alpha$-$\mu$ RVs are derived. Section~IV presents some application uses of the derived expressions, while Section~V shows  some  illustrative numerical results  and  draws  some  discussions. Finally, some concluding remarks are presented in Section~VI.

In what follows, we use the following notation: $f_{A}(\cdot)$ and $F_{A}(\cdot)$ for the PDF and CDF of a RV~$A$, respectively; $\mathbb{E}[\cdot]$ for expectation; $\mathbb{V}[\cdot]$ for variance; $\Pr\left \{ \cdot  \right \}$ for probability; and $\abs{\cdot}$ for absolute value. In addition, $\Gamma(\cdot)$ is the gamma function~\cite[Eq.~(6.1.1)]{Abramowitz}; $\operatorname{P}(z,y)=\tfrac{1}{\Gamma(z)} \int_{0}^{y}t^{z-1}\text{exp}(-t)dt$ is the regularized lower incomplete gamma function~\cite[Eq.~(6.5.1)]{Abramowitz}; $\mathrm{H}_{p,q}^{m,n}\left[ \cdot
\right]$ is the Fox H-function~\cite[Eq.~(1.1)]{Fox}; and $G_{p, q}^{m, n}\left[\cdot
\right]$ is the Meijer G-function~\cite[Eq.~(7.82)]{Gradshteyn}. We also use $\mathrm{i}=\sqrt[]{-1}$ for the imaginary unit; $\field{N}^0$ for
natural numbers including zero; $\mathbb{C}$ for complex numbers; $\mathbb{R}$ for real numbers; $\mathbb{R}^+$ for positive real numbers; $\approx$ to denote ``approximately equal~to''; and $\propto$ to denote ``proportionally~to''.

\section{Preliminaries}
The PDF of the envelope $R$ of a signal propagating on a fading channel with distribution $\alpha$-$\mu$ is given by~\cite{yacoub2007alpha} 

\begin{equation}\label{eq:pdfalpha}
f_{R}(r)=\frac{\alpha\mu^{\mu}r^{\alpha\mu-1}}{\hat{r}^{\mu \alpha}\Gamma (\mu)}\exp\left(-\frac{\mu r^{\alpha}}{\hat{r}^{\alpha}} \right ),
\end{equation}
where $\alpha$ denotes the non-linearity of the environment, $\hat{r}=\sqrt[\alpha]{\mathbb{E}\left [ R^{\alpha} \right ]}$ is the $\alpha$-root mean value of the channel envelope, and
$\mu=\hat{r}^{2\alpha}\mathbb{V}^{-1}\left [ R^{\alpha} \right ]$ is related to the number of multipath clusters. 
%
Therefore, some special cases for the parameters $\alpha$ and $\mu$, such that the $\alpha$-$\mu$ distribution reduces to well-known distributions commonly used in wireless application scenarios, are specified in Table~\ref{specialcases}~\cite{yacoub2007alpha}.

%
\begin{table}[H]
\scriptsize
		\centering
    \caption{Particular cases of the $\alpha$-$\mu$ distribution}
	\centering
	\begin{tabular}{cc}
		\toprule
        \hspace{1mm}
		\textbf{Distribution} &  \hspace{1mm}     \textbf{$\alpha$-$\mu$ fading values }  \\ \cmidrule(lr){1-2}
		\multicolumn{1}{l}
  \textbf{Nakagami-$m$} &  \hspace{3mm}     \textbf{$\alpha=2$, $\mu=m$  }  \\ \cmidrule(lr){1-2}
       \multicolumn{1}{l}
        \textbf{Weibull}  \hspace{3mm}  &  \textbf{ $\alpha=\alpha$, $\mu=1$}  \\ \cmidrule(lr){1-2} 
        \multicolumn{1}{l}
        \textbf{Rayleigh}  \hspace{2.5mm} &  \hspace{2mm}     \textbf{$\alpha=2$, $\mu=1$ }  \\ \cmidrule(lr){1-2}\end{tabular}\label{specialcases}
\end{table}


From~\eqref{eq:pdfalpha}, the $n$-th moment
$\mathbb{E}\left [ R^n \right ]$ can be obtained as 
\begin{equation}\label{eq:moments}
\mathbb{E}\left [ R^n \right ]=\hat{r}^{n} \frac{\Gamma\left ( \mu+n/\alpha  \right )}{\mu^{n/\alpha}  \Gamma (\mu)}.
\end{equation}

Let $\Upsilon \stackrel{\Delta}{=} \gamma_t R^2$ be the instantaneous received signal-to-noise ratio (SNR) through an $\alpha$-$\mu$ fading channel, with $\gamma_t$ being the transmit SNR~\cite{DaCosta,Lei2017}. Hence, the corresponding PDF and CDF can be obtained from~\eqref{eq:pdfalpha} by performing a transformation of variables as in~\cite[Eqs.~(8)~and~(10)]{DaCosta}
\begin{align}
f_{\Upsilon}(\gamma) & =\frac{\alpha \gamma^{(\alpha\mu/2)-1}}{2\beta^{\alpha\mu/2}\Gamma (\mu)}\exp\left[-\left ( \frac{\gamma}{\beta}  \right )^{\alpha/2}\right ],\label{eq:2}\\
F_{\Upsilon}(\gamma) & = \operatorname{P} \left ( \mu, \left ( \frac{\gamma}{\beta}  \right ) ^{\alpha/2} \right ),\label{eq:3}
\end{align}
where $\beta=\bar{\Upsilon}\Gamma(\mu) /\Gamma(\mu+2/\alpha)$, with $\bar{\Upsilon}$ being the average received SNR, so that
\begin{align}\label{eq:4}
\bar{\Upsilon} & =\mathbb{E}\left [\Upsilon\right ]\nonumber\\
& =\hat{r}^{2} \frac{\Gamma(\mu+2/\alpha)}{\mu^{2/\alpha}\Gamma(\mu)}\gamma_t,
\end{align}
Now, by using~\cite[Eq. (01.03.26.0004.01)]{Wolfram1}, we can express the exponential function in~\eqref{eq:2} in terms of the Meijer G-function, so that the PDF  of $\Upsilon$ can be rewritten as
\begin{equation}\label{eq:6}
f_{\Upsilon}(\gamma)=\frac{\alpha \gamma^{(\alpha\mu/2)-1}}{2\beta^{\alpha\mu/2}\Gamma (\mu)} G_{0,1}^{1,0}\left[ \left ( \frac{\gamma}{\beta}  \right )^{\alpha/2} \bigg|
\begin{array}{c}
 0 \\
\end{array}
\right].
\end{equation}
Likewise, using~\cite[Eq. (06.09.26.0006.01)]{Wolfram1}, the regularized lower incomplete gamma function in~\eqref{eq:3} can be expressed in terms of the Meijer G-function. Thus, the CDF of $\Upsilon$
can be rewritten as
 \begin{equation}\label{eq:7}
F_{\Upsilon}(\gamma)=\frac{1}{\Gamma(\mu)} \left ( \frac{\gamma}{\beta}  \right ) ^{\frac{\mu\alpha}{2}} G_{1,2}^{1,1}\left[ \left ( \frac{\gamma}{\beta}  \right ) ^{\frac{\alpha}{2}} \bigg|
\begin{array}{c}
 1-\mu\\
 0,-\mu \\
\end{array}
\right].
\end{equation}

\section{Statistics of the ratio of independent and arbitrary squared $\alpha$-$\mu$ RVs }
In this section, we derive closed-form expressions for the PDF, CDF, MGF and higher order moments of the ratio $X$$=$$\Upsilon_1/\Upsilon_2 $, where $\Upsilon_1$ and $\Upsilon_1$ are i.n.i.d. RVs following an $\alpha$-$\mu$ distribution. Moreover, hereafter we assume that $\alpha_1, \alpha_2 \in \mathbb{R}^+$, $k = \tfrac{\alpha_1}{\alpha_2}$, $\mu_1, \mu_2 \in \mathbb{R}^+$, and $x \in \mathbb{R}^+$.
\subsection{PDF, CDF and MGF of $X$}
Herein, the PDF and CDF of the ratio of two independent squared $\alpha$-$\mu$ RVs are given in the following proposition. Besides, as one of the most important characterizations of a RV, the corresponding MGF is also provided.
\begin{prop}\label{prop:pdf}
Let $\Upsilon_1$ and $\Upsilon_2$ be i.n.i.d. squared $\alpha$-$\mu$ distributed RVs with probability functions given as in~\eqref{eq:6} and~\eqref{eq:7}. The PDF, CDF, and MGF of the ratio $X {=}\Upsilon_1/\Upsilon_2 $ are respectively given~by
\begin{align}
f_X(x) = & \frac{\alpha_{1} x^{\frac{\alpha_1\mu_1}{2}-1}\beta_2^{\frac{\alpha_1\mu_1}{2} }}{2 \beta_1^{\frac{\alpha_1\mu_1}{2}}\Gamma (\mu_2)\Gamma (\mu_1)}  \nonumber \\
&\times \underset{\mathrm{H}_1}{\underbrace{ \mathrm{H}_{1,1}^{1,1}\left[{\left ( \frac{x\beta_2}{\beta_1}  \right )^{\frac{\alpha_1}{2}}}\bigg|
\begin{array}{c}
 (1-\mu_2-k\mu_1,k)\\
 (0,1)\\
\end{array}
\right]}},\label{pdfRatio}\\
F_X(x) = & \frac{1}{\Gamma(\mu_2)\Gamma (\mu_1)} \left ( \frac{x\beta_2}{\beta_1}  \right )^{\frac{\alpha_1\mu_1}{2}}\nonumber \\ 
&\times \underset{\mathrm{H}_2}{\underbrace{ \mathrm{H}_{2,2}^{1,2}\left[{\left ( \!\frac{x\beta_2}{\beta_1}\!  \!\right)^{\!\frac{\alpha_1}{2}}}\bigg|
\begin{array}{c}
 \!(1\!-\!\mu_1,1\!),\!(1\!-\!\mu_1 k\!-\! \mu_2,k)\! \\
 (0,1) ,\hspace{0.5mm} \!(-\mu_1,1)\!\\
\end{array}
\right]}},\label{eq:CDFRATIO}\\
{\cal M}_X(s) = & \frac{\alpha_{1}   }{2 \Gamma (\mu_2)\Gamma (\mu_1)} \left ( \frac{\beta_2}{s\beta_1} \right )^{\frac{\alpha_1\mu_1}{2}} \nonumber \\
\times &\underset{\mathrm{H}_3}{\underbrace{ \mathrm{H}_{2,1}^{1,2}\!\left[{\left ( \frac{\beta_2}{s\beta_1}  \right )^{\frac{\alpha_1}{2}}}\!\bigg|\!\!\!
\begin{array}{c}
 (1\!-\!\mu_2\!-\!k\mu_1,k),(1-\frac{\mu_1 \alpha_1}{2},\frac{\alpha_1}{2}) \\
 (0,1) \\
\end{array}
\!\!\right]}}.\label{eq:MGF}
\end{align}
\end{prop}
\begin{proof}
	See Appendix~\ref{ap:statistics}.
\end{proof}
\begin{remark}
Notice that contrary to previous works~\cite{Leonardo,Leonardo2016}, the results of Proposition~\ref{prop:pdf} are general, since no constraints are imposed on the parameters of $\Upsilon_1$ and $\Upsilon_2$.
\end{remark}

\begin{remark}
It is worth mentioning that currently the Fox H-function is not implemented in mathematical software packages such as Wolfram Mathematica. However, the  Fox H-function can be evaluated using either numerical evaluations in the form of a Mellin--Barnes integral\footnote{In Appendix~\ref{ap:mathimplementation}, we provide a portable implementation of the Fox H-function in MATHEMATICA\textregistered Wolfram. The code is simple, efficient, and provides very accurate results.}~\cite{Fox} or by applying calculus of residues\footnote{ 
An alternative method to compute the results presented
here is given by the series representation for the Fox H-functions $\mathrm{H}_1$, $\mathrm{H}_2$, and $\mathrm{H}_3$ as in~\eqref{eq:FoxbyResidues1},~\eqref{eq:FoxbyResidues2} and~\eqref{eq:FoxbyResidues3}, respectively, shown at the bottom of next page. The mathematical derivation of the referred expressions is provided in Appendix~\ref{ap:residues}}.
\end{remark}

\subsection{Higher Order Moments}
The $n$th order moment for a RV $X$ is defined as $\mathbb{E}\left [X^n\right ]\buildrel \Delta \over = \int_{0}^{\infty}x^{n}f_X(x)dx$. Then, to calculate the $n$th moment of the ratio of squared $\alpha${-}$\mu$ distributed RVs, $X{=}\Upsilon_1/\Upsilon_2$, we resort to the identity for the product of two statistically independent  RVs, i.e., $\mathbb{E}\left [(\Upsilon_1\Upsilon_2)^n\right ]=\mathbb{E}\left [\Upsilon_1^n\right ] \mathbb{E}\left [\Upsilon_2^n\right ]$~\cite{productCarlos}.
However, for the case of the ratio of two RVs, we are interested in solving $\mathbb{E}\left [(\Upsilon_1/\Upsilon_2)^n\right ]$, thus being necessary to determine the $n$th moment of the inverse of a RV. To this end, let us define $Z=1/\Upsilon_2$, such that $\mathbb{E}\left [(\Upsilon_1Z)^n\right ]=\mathbb{E}\left [\Upsilon_1^n\right ] \mathbb{E}\left [Z^{n}\right ]$. Thus, by determining the moments of $\mathbb{E}\left [\Upsilon_1^n\right ]$ and $\mathbb{E}\left [Z^{n}\right ]$, then the higher order moments of $\mathbb{E}\left [X^n\right ]$ can be found. The moments of $Z$ are determined from the distribution of the inverse of $R_2$ by considering $\mathbb{E}\left [(\gamma_t R_2^2)^{n}\right ]=\mathbb{E}\left [\Upsilon_2^{n}\right ]$~\cite{DaCosta}. From this consideration, the higher order moments $\mathbb{E}\left [X^n\right ]$ can be obtained as in the following proposition. 

\begin{prop}\label{prop:moments}
The $n$th order moment for the  ratio of squared $\alpha${-}$\mu$ distributed RVs, $X=\Upsilon_1/\Upsilon_2$, is given by
\begin{align}\label{eq:HigherMoments}
\mathbb{E}\left [X^n\right ]=& \frac{\left (\hat{r_1}\hat{r_2}\right )^{2n}\Gamma\left ( \mu_1+\frac{2n}{\alpha_1} \right )\Gamma\left ( \mu_2-\frac{2n}{\alpha_2} \right )}{\mu_1^{2n/\alpha_1} \mu_2^{2n/\alpha_2}\Gamma\left ( \mu_1 \right )\Gamma\left ( \mu_2\right )},
\nonumber \\ & 
\hspace{12mm} for \hspace{2mm} n>\mu_i \alpha_i,\hspace{2mm} i \in \left \{1,2.\right \}.
\end{align}
\end{prop}
\vspace{2mm}
\begin{proof}
	See Appendix~\ref{ap:momentinv}.
\end{proof}
\begin{remark}
An equivalent expression for the $n$th order moment in~\eqref{eq:HigherMoments} can be obtained by applying the Mellin transform~\cite[Eq. (6.3.3.c)]{springer} to the PDF in~\eqref{pdfRatio}. 
\end{remark}

The formulations derived in~\eqref{pdfRatio} to~\eqref{eq:HigherMoments} are general results that can be reduced to other distributions for
different channel models, such as Rayleigh, Nakagami-$m$, and
Weibull, by considering the corresponding parameters as in Table~\ref{specialcases}. Therefore, the PDF, CDF, and MGF for the distribution of the ratio of the aforementioned distributions are given in Table~\ref{RATIOPDF},~\ref{RATIOCDF} and~\ref{RATIOMGF}, respectively.  
\begin{table*}[t]
\scriptsize
		\centering
    \caption{PDF of the Ratio for Different Distributions as Special Cases}
	\centering
	\begin{tabular}{ll}
		\toprule
        \hspace{10mm}
		\textbf{Ratio} &  \hspace{30mm}     \textbf{PDF }  \\
		\cmidrule(lr){1-2}
		\multicolumn{1}{l}{Nakagami-$m$$\mathlarger{\mathlarger{\mathlarger{/}}}$Nakagami-$m$}& 
$\begin{array} {lcl} f_X(x)=\frac{ x^{\mu_1-1}\beta_2^{\mu_1 }}{ \beta_1^{\mu_1}\Gamma (\mu_2)\Gamma (\mu_1)}  G_{1,1}^{1,1}\left[ \frac{x\beta_2}{\beta_1} \bigg|
\begin{array}{c}
 1-\mu_2-\mu_1\\
 0\\
\end{array}
\right] \end{array}$
\\	\cmidrule(lr){1-2}
        \multicolumn{1}{l}{Nakagami-$m$$\mathlarger{\mathlarger{\mathlarger{/}}}$Weibull}&$\begin{array} {lcl}  f_X(x)=\frac{x^{\mu_1-1}\beta_2^{\mu_1 }}{ \beta_1^{\mu_1}\Gamma (\mu_2)} \mathrm{H}_{1,1}^{1,1}\left[ \frac{x\beta_2}{\beta_1} \bigg|
\begin{array}{c}
 (- \frac{2\mu_1}{\alpha_2}, \frac{2}{\alpha_2})\\
 (0,1)\\
\end{array}
\right] \end{array}$  \\
		\cmidrule(lr){1-2}
        \multicolumn{1}{l}{Nakagami-$m$$\mathlarger{\mathlarger{\mathlarger{/}}}$Rayleigh}&$\begin{array} {lcl}  f_X(x)=\frac{x^{\mu_1-1}\beta_2^{\mu_1 }}{ \beta_1^{\mu_1}\Gamma (\mu_1)} G_{1,1}^{1,1}\left[ \frac{x\beta_2}{\beta_1} \bigg|
\begin{array}{c}
 -\mu_1\\
 0\\
\end{array}
\right] \end{array}$  \\
		\cmidrule(lr){1-2}
        \multicolumn{1}{l}{Weibull$\mathlarger{\mathlarger{\mathlarger{/}}}$Weibull}& $\begin{array} {lcl}  f_X(x)=\frac{\alpha_{1} x^{\frac{\alpha_1}{2}-1}\beta_2^{\frac{\alpha_1}{2} }}{2 \beta_1^{\frac{\alpha_1}{2}}} \mathrm{H}_{1,1}^{1,1}\left[\left ( \frac{x\beta_2}{\beta_1}  \right )^{\frac{\alpha_1}{2}}\bigg|
\begin{array}{c}
 (-k,k)\\
 (0,1)\\
\end{array}
\right] \end{array}$   \\
		\cmidrule(lr){1-2}
        \multicolumn{1}{l}{Weibull$\mathlarger{\mathlarger{\mathlarger{/}}}$Nakagami-$m$}& $\begin{array} {lcl}  f_X(x)=\frac{\alpha_{1} x^{\frac{\alpha_1}{2}-1}\beta_2^{\frac{\alpha_1}{2} }}{2 \beta_1^{\frac{\alpha_1}{2}}\Gamma (\mu_2)} \mathrm{H}_{1,1}^{1,1}\left[\left ( \frac{x\beta_2}{\beta_1}  \right )^{\frac{\alpha_1}{2}}\bigg|
\begin{array}{c}
 (1-\mu_2- \frac{\alpha_1}{2}, \frac{\alpha_1}{2})\\
 (0,1)\\
\end{array}
\right] \end{array}$   \\
		\cmidrule(lr){1-2}
        \multicolumn{1}{l}{Weibull$\mathlarger{\mathlarger{\mathlarger{/}}}$Rayleigh}& $\begin{array} {lcl}  f_X(x)=\frac{\alpha_{1} x^{\frac{\alpha_1}{2}-1}\beta_2^{\frac{\alpha_1}{2} }}{2 \beta_1^{\frac{\alpha_1}{2}}} \mathrm{H}_{1,1}^{1,1}\left[\left ( \frac{x\beta_2}{\beta_1}  \right )^{\frac{\alpha_1}{2}}\bigg|
\begin{array}{c}
 (- \frac{\alpha_1}{2}, \frac{\alpha_1}{2})\\
 (0,1)\\
\end{array}
\right] \end{array}$   \\
		\cmidrule(lr){1-2}
        \multicolumn{1}{l}{Rayleigh$\mathlarger{\mathlarger{\mathlarger{/}}}$Rayleigh}& $\begin{array} {lcl}  f_X(x)=\frac{\beta_2}{ \beta_1} G_{1,1}^{1,1}\left[ \frac{x\beta_2}{\beta_1} \bigg|
\begin{array}{c}
 -1\\
 0\\
\end{array}
\right] \end{array}$   \\
		\cmidrule(lr){1-2}
        \multicolumn{1}{l}{Rayleigh$\mathlarger{\mathlarger{\mathlarger{/}}}$Nakagami-$m$}& $\begin{array} {lcl}  f_X(x)=\frac{\beta_2}{ \beta_1 \Gamma(\mu_2)} G_{1,1}^{1,1}\left[ \frac{x\beta_2}{\beta_1} \bigg|
\begin{array}{c}
 -\mu_2\\
 0\\
\end{array}
\right] \end{array}$   \\
		\cmidrule(lr){1-2}
        \multicolumn{1}{l}{Rayleigh$\mathlarger{\mathlarger{\mathlarger{/}}}$Weibull}&  $\begin{array} {lcl}  f_X(x)=\frac{ \beta_2}{ \beta_1} \mathrm{H}_{1,1}^{1,1}\left[ \frac{x\beta_2}{\beta_1} \bigg|
\begin{array}{c}
 (- \frac{2}{\alpha_2}, \frac{2}{\alpha_2})\\
 (0,1)\\
\end{array}
\right] \end{array}$  \\
		\cmidrule(lr){1-2}
	\end{tabular}\label{RATIOPDF}
\end{table*}

\begin{table*}[t]
\scriptsize
		\centering
    \caption{CDF of the Ratio for Different Distributions as Special Cases}
	\centering
	\begin{tabular}{ll}
		\toprule
\hspace{10mm}		\textbf{Ratio} &     \hspace{40mm} \textbf{CDF }  \\
		\cmidrule(lr){1-2}
		\multicolumn{1}{l}{Nakagami-$m$$\mathlarger{\mathlarger{\mathlarger{/}}}$Nakagami-$m$}& 
$\begin{array} {lcl} F_X(x)=\frac{1}{\Gamma(\mu_2)\Gamma (\mu_1)} \left ( \frac{x\beta_2}{\beta_1}  \right )^{\mu_1}G_{2,2}^{1,2}\left[ \frac{x\beta_2}{\beta_1}  \bigg|
\begin{array}{c}
 1-\mu_1,1-\mu_1-\mu_2 \\
 0 ,\hspace{0.5mm} -\mu_1\\
\end{array}
\right] \end{array}$
\\	\cmidrule(lr){1-2}
        \multicolumn{1}{l}{Nakagami-$m$$\mathlarger{\mathlarger{\mathlarger{/}}}$Weibull}&$\begin{array} {lcl} F_X(x)=\frac{1}{\Gamma (\mu_1)} \left ( \frac{x\beta_2}{\beta_1}  \right )^{\mu_1}\mathrm{H}_{2,2}^{1,2}\left[ \frac{x\beta_2}{\beta_1} \bigg|
\begin{array}{c}
 (1-\mu_1,1),(-\frac{2\mu_1}{\alpha_2},\frac{2}{\alpha_2}) \\
 (0,1) ,\hspace{0.5mm} (-\mu_1,1)\\
\end{array}
\right] \end{array}$  \\
		\cmidrule(lr){1-2}
        \multicolumn{1}{l}{Nakagami-$m$$\mathlarger{\mathlarger{\mathlarger{/}}}$Rayleigh}&$\begin{array} {lcl} F_X(x)=\frac{1}{\Gamma (\mu_1)} \left ( \frac{x\beta_2}{\beta_1}  \right )^{\mu_1}G_{2,2}^{1,2}\left[ \frac{x\beta_2}{\beta_1} \bigg|
\begin{array}{c}
 1-\mu_1,-\mu_1 \\
 0,\hspace{0.5mm} -\mu_1\\
\end{array}
\right] \end{array}$  \\
		\cmidrule(lr){1-2}
        \multicolumn{1}{l}{Weibull$\mathlarger{\mathlarger{\mathlarger{/}}}$Weibull}& $\begin{array} {lcl} F_X(x)=\left ( \frac{x\beta_2}{\beta_1}  \right )^{\frac{\alpha_1}{2}}\mathrm{H}_{2,2}^{1,2}\left[\left ( \frac{x\beta_2}{\beta_1}  \right )^{\frac{\alpha_1}{2}}\bigg|
\begin{array}{c}
 (0,1),(-k,k) \\
 (0,1) ,\hspace{0.5mm} (-1,1)\\
\end{array}
\right] \end{array}$   \\
		\cmidrule(lr){1-2}
        \multicolumn{1}{l}{Weibull$\mathlarger{\mathlarger{\mathlarger{/}}}$Nakagami-$m$}& $\begin{array} {lcl} F_X(x)=\frac{1}{\Gamma(\mu_2)} \left ( \frac{x\beta_2}{\beta_1}  \right )^{\frac{\alpha_1}{2}}\mathrm{H}_{2,2}^{1,2}\left[\left ( \frac{x\beta_2}{\beta_1}  \right )^{\frac{\alpha_1}{2}}|
\begin{array}{c}
 (0,1),(1-k-\mu_2,k) \\
 (0,1),\hspace{0.5mm} (-1,1)\\
\end{array}
\right] \end{array}$ \\
		\cmidrule(lr){1-2}
        \multicolumn{1}{l}{Weibull$\mathlarger{\mathlarger{\mathlarger{/}}}$Rayleigh}& $\begin{array} {lcl} F_X(x)=\left ( \frac{x\beta_2}{\beta_1}  \right )^{\frac{\alpha_1}{2}}\mathrm{H}_{2,2}^{1,2}\left[\left ( \frac{x\beta_2}{\beta_1}  \right )^{\frac{\alpha_1}{2}}\bigg|
\begin{array}{c}
 (0,1),(-\frac{\alpha_1}{2} ,\frac{\alpha_1}{2}) \\
 (0,1) ,\hspace{0.5mm} (-1,1)\\
\end{array}
\right] \end{array}$  \\
		\cmidrule(lr){1-2}
        \multicolumn{1}{l}{Rayleigh$\mathlarger{\mathlarger{\mathlarger{/}}}$Rayleigh}& $\begin{array} {lcl} F_X(x)=  \frac{x\beta_2}{\beta_1}  G_{2,2}^{1,2}\left[ \frac{x\beta_2}{\beta_1} \bigg|
\begin{array}{c}
 0,-1 \\
0,-1\\
\end{array}
\right] \end{array}$   \\
		\cmidrule(lr){1-2}
        \multicolumn{1}{l}{Rayleigh$\mathlarger{\mathlarger{\mathlarger{/}}}$Nakagami-$m$}& $\begin{array} {lcl} F_X(x)=\frac{x\beta_2}{\beta_1\Gamma(\mu_2)} G_{2,2}^{1,2}\left[\frac{x\beta_2}{\beta_1} \bigg|
\begin{array}{c}
 0,-\mu_2 \\
 0,\hspace{0.5mm}-1\\
\end{array}
\right] \end{array}$  \\
		\cmidrule(lr){1-2}
        \multicolumn{1}{l}{Rayleigh$\mathlarger{\mathlarger{\mathlarger{/}}}$Weibull}&  $\begin{array} {lcl} F_X(x)= \frac{x\beta_2}{\beta_1}  \mathrm{H}_{2,2}^{1,2}\left[\frac{x\beta_2}{\beta_1}  \bigg|
\begin{array}{c}
 (0,1),(-\frac{2}{\alpha_2},\frac{2}{\alpha_2}) \\
 (0,1) ,\hspace{2.5mm} (-1,1)\\
\end{array}
\right] \end{array}$  \\
		\cmidrule(lr){1-2}
	\end{tabular}\label{RATIOCDF}
\end{table*}

\begin{table*}[t]
\scriptsize
	\centering
    \caption{MGF of the Ratio for Different Distributions as Special Cases}
	\centering
	\begin{tabular}{ll}
		\toprule
\hspace{10mm}		\textbf{Ratio} &     \hspace{40mm} \textbf{MGF}  \\
		\cmidrule(lr){1-2}
		\multicolumn{1}{l}{Nakagami-$m$$\mathlarger{\mathlarger{\mathlarger{/}}}$Nakagami-$m$}& 
$\begin{array} {lcl}  {\cal M}_X(s)= \frac{1 }{2 \Gamma (\mu_2)\Gamma (\mu_1)} \left ( \frac{\beta_2}{s\beta_1} \right )^{\mu_1} G_{2,1}^{1,2}\left[\frac{\beta_2}{s\beta_1}\bigg|
\begin{array}{c}
 1-\mu_2-\mu_1,1-\mu_1 \\
 0 \\
\end{array}
\right]\end{array}$
\\	\cmidrule(lr){1-2}
        \multicolumn{1}{l}{Nakagami-$m$$\mathlarger{\mathlarger{\mathlarger{/}}}$Weibull}&$\begin{array} {lcl} {\cal M}_X(s)= \frac{1 }{ \Gamma (\mu_1)} \left ( \frac{\beta_2}{s\beta_1} \right )^{\mu_1} \mathrm{H}_{2,1}^{1,2}\left[\frac{\beta_2}{s\beta_1}  \bigg|
\begin{array}{c}
 (-\frac{2\mu_1}{\alpha_2},\frac{2}{\alpha_2}),(1-\mu_1,1) \\
 (0,1) \\
\end{array}
\right]\end{array}$ \\
		\cmidrule(lr){1-2}
        \multicolumn{1}{l}{Nakagami-$m$$\mathlarger{\mathlarger{\mathlarger{/}}}$Rayleigh}&$\begin{array} {lcl} {\cal M}_X(s)= \frac{1}{ \Gamma (\mu_1)}\left ( \frac{\beta_2}{s\beta_1} \right )^{\mu_1}G_{2,1}^{1,2}\left[\frac{\beta_2}{s\beta_1} \bigg|
\begin{array}{c}
 -\mu_1,1-\mu_1 \\
 0 \\
\end{array}
\right]\end{array}$ \\
		\cmidrule(lr){1-2}
        \multicolumn{1}{l}{Weibull$\mathlarger{\mathlarger{\mathlarger{/}}}$Weibull}& $\begin{array} {lcl} {\cal M}_X(s)= \frac{\alpha_{1}  }{2} \left ( \frac{\beta_2}{s\beta_1} \right )^{\frac{\alpha_1}{2}}\mathrm{H}_{2,1}^{1,2}\left[\left ( \frac{\beta_2}{s\beta_1}  \right )^{\frac{\alpha_1}{2}}|
\begin{array}{c}
 (-k,k),(1-\frac{ \alpha_1}{2},\frac{\alpha_1}{2}) \\
 (0,1) \\
\end{array}
\right]\end{array}$  \\
		\cmidrule(lr){1-2}
        \multicolumn{1}{l}{Weibull$\mathlarger{\mathlarger{\mathlarger{/}}}$Nakagami-$m$}& $\begin{array} {lcl} {\cal M}_X(s)= \frac{\alpha_{1}  }{2 \Gamma (\mu_2)} \left ( \frac{\beta_2}{s\beta_1} \right )^{\frac{\alpha_1}{2}}\mathrm{H}_{2,1}^{1,2}\left[\left ( \frac{\beta_2}{s\beta_1}  \right )^{\frac{\alpha_1}{2}}|
\begin{array}{c}
 (1-\mu_2-\frac{\alpha_1}{2} ,\frac{\alpha_1}{2} ),(1-\frac{
 \alpha_1}{2},\frac{\alpha_1}{2}) \\
 (0,1) \\
\end{array}
\right]\end{array}$ \\
		\cmidrule(lr){1-2}
        \multicolumn{1}{l}{Weibull$\mathlarger{\mathlarger{\mathlarger{/}}}$Rayleigh}& $\begin{array} {lcl} {\cal M}_X(s)= \frac{\alpha_{1}  }{2 } \left ( \frac{\beta_2}{s\beta_1} \right )^{\frac{\alpha_1}{2}}\mathrm{H}_{2,1}^{1,2}\left[\left ( \frac{\beta_2}{s\beta_1}  \right )^{\frac{\alpha_1}{2}}\bigg|
\begin{array}{c}
 (-\frac{\alpha_1}{2},\frac{\alpha_1}{2}),(1-\frac{ \alpha_1}{2},\frac{\alpha_1}{2}) \\
 (0,1) \\
\end{array}
\right]\end{array}$ \\
		\cmidrule(lr){1-2}
        \multicolumn{1}{l}{Rayleigh$\mathlarger{\mathlarger{\mathlarger{/}}}$Rayleigh}& $\begin{array} {lcl} {\cal M}_X(s)=  \frac{\beta_2}{s\beta_1} G_{2,1}^{1,2}\left[\frac{\beta_2}{s\beta_1}  \bigg|
\begin{array}{c}
 -1,0 \\
 0 \\
\end{array}
\right]\end{array}$  \\
		\cmidrule(lr){1-2}
        \multicolumn{1}{l}{Rayleigh$\mathlarger{\mathlarger{\mathlarger{/}}}$Nakagami-$m$}& $\begin{array} {lcl} {\cal M}_X(s)= \frac{\beta_2   }{ s \beta_1\Gamma (\mu_2)}G_{2,1}^{1,2}\left[ \frac{\beta_2}{s\beta_1} \bigg|
\begin{array}{c}
 -\mu_2,0 \\
 0\\
\end{array}
\right]\end{array}$ \\
		\cmidrule(lr){1-2}
        \multicolumn{1}{l}{Rayleigh$\mathlarger{\mathlarger{\mathlarger{/}}}$Weibull}&  $\begin{array} {lcl} {\cal M}_X(s)= \frac{\beta_2}{s\beta_1} \mathrm{H}_{2,1}^{1,2}\left[\frac{\beta_2}{s\beta_1}  \bigg|
\begin{array}{c}
 (-\frac{2}{\alpha_2},\frac{2}{\alpha_2}),(0,1) \\
 (0,1) \\
\end{array}
\right]\end{array}$  \\
		\cmidrule(lr){1-2}
	\end{tabular}\label{RATIOMGF}
\end{table*}


\begin{figure*}[hbt]
	\begin{footnotesize}
   \begin{equation}
     \mathrm{H}_1=\begin{cases}
       \sum_{h=0}^{\infty}\frac{z^{h}\Gamma\left (k (h+\mu_1)+\mu_2\right )}{\left (-1\right )^h  h!}, & \text{$k\leq 1$, if $k=1\rightarrow $  $\abs{z}<1$}.\\
        \sum_{h=0}^{\infty}\frac{ z^{-\frac{h+k\mu_1+\mu_2}{k}}\Gamma\left (\frac{h+k\mu_1+\mu_2}{k} \right )}{\left (-1\right )^h  k h!}, & \text{$k\geq 1$, if $k=1\rightarrow $  $\abs{z}>1$}.
         \end{cases}
\label{eq:FoxbyResidues1}
		\end{equation}
		\end{footnotesize}
	\hrulefill
\end{figure*}

\begin{figure*}[hbt]
	\begin{footnotesize}
   \begin{equation}
     \mathrm{H}_2=\begin{cases}
  \sum_{h=0}^{\infty}\frac{z^{h}\Gamma\left (k(h+\mu_1)+\mu_2\right )}{\left (-1\right )^h  (h+\mu_1) \Gamma\left ( 1+h \right )}, & \text{$k\leq 1$, if $k=1\rightarrow $  $\abs{z}<1$}.\\
        \sum_{h=0}^{\infty}\frac{z^{-h-\mu_1} \Gamma\left (h+\mu_1\right )\Gamma\left (-hk+\mu_2\right )}{\left (-1\right )^{h-2} \Gamma \left ( 1-h \right )h!}+ \sum_{h=0}^{\infty}\frac{z^{-\frac{h}{k}-\mu_1-\frac{\mu_2}{k}}\Gamma\left (\frac{-h-\mu_2}{k}\right )\Gamma\left (\frac{h+k\mu_1+\mu_2}{k}  \right )}{\left (-1\right )^{h-2} \Gamma \left ( \frac{-h+k-\mu_2}{k} \right ) kh!}, & \text{$k\geq 1$, if $k=1\rightarrow $  $\abs{z}>1$}.
         \end{cases}
\label{eq:FoxbyResidues2}
		\end{equation}
		\end{footnotesize}
	\hrulefill
\end{figure*}

\begin{figure*}[hbt]
	\begin{footnotesize}
   \begin{equation}
     \mathrm{H}_3=\begin{cases}
  \sum_{h=0}^{\infty}\frac{\Gamma\left (hk+k\mu_X+\mu_Y\right )z^{h}}{\left (-1\right )^h  h!}, & \text{$k\leq 1$, if $k=1\rightarrow $  $\abs{z}<1$}.\\
        \sum_{h=0}^{\infty}\frac{\Gamma\left (\frac{h+k\mu_X+\mu_Y}{k}\right )z^{-\frac{h+k\mu_X+\mu_Y}{k}}}{\left (-1\right )^h k h!}, & \text{$k\geq 1$, if $k=1\rightarrow $  $\abs{z}>1$}.
         \end{cases}
\label{eq:FoxbyResidues3}
		\end{equation}
		\end{footnotesize}
	\hrulefill
\end{figure*}

\section{Applications}

In this section, we present some illustrative application uses of our analytical expressions in the context of key enabling technologies for 5G and beyond networks, including PLS, CR, and FD relaying.

\subsection{Physical Layer Security and Secrecy Outage Probability}
In the context of physical layer security, a widely used metric to evaluate the secrecy performance of wireless networks is the secrecy outage probability. Thus, let us consider the Wyner's wiretap channel as depicted in Fig.~\ref{sistema1}, where a legitimate transmitter (Alice) sends confidential messages to the legitimate receiver (Bob) through the main channel, while the eavesdropper (Eve) tries to intercept these messages from its received signal over the eavesdropper channel. Furthermore, assume that both the main and eavesdropper channels experience independent $\alpha$-$\mu$ distributed fading.

\begin{figure}[H]
\centering 
\psfrag{A}[Bc][Bc][0.8]{A}
\psfrag{B}[Bc][Bc][0.8]{B}
\psfrag{E}[Bc][Bc][0.8]{E}
\psfrag{U}[Bc][Bc][0.8]{$h_{\mathrm{AB}}$}
\psfrag{w}[Bc][Bc][0.8][-20]{$h_{\mathrm{AE}}$}
\psfrag{Main channel}[Bc][Bc][0.6]{Main channel} 
\psfrag{Wiretap channel}[Bc][Bc][0.6]{Wiretap channel} 
\includegraphics[width=0.7\linewidth]{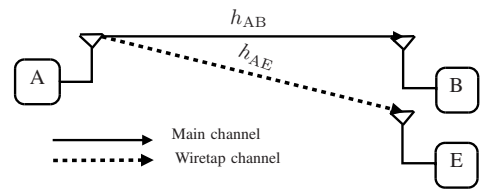} \caption{The system model of a   wiretap channel consisting of two legitimate correspondents and one eavesdropper.}
\label{sistema1}
\end{figure}

According to~\cite{Wyner}, the secrecy capacity is
obtained as
\begin{align}\label{eq:8}
C_s&=\!\text{max}\left \{C_B-C_E,0  \right \} \nonumber \\
&=\!\text{max}\left \{\log_2\!\left (1\!+\!\frac{|h_{\mathrm{AB}}|^2P_{\mathrm{A}}}{N_{\mathrm{0}}} \!\right )\!-\!\log_2\!\left (1\!+\!\frac{|h_{\mathrm{AE}}|^2P_{\mathrm{A}}}{N_{\mathrm{0}}}\!\right ),0  \right \} \nonumber \\ 
&=\!\text{max}\left \{\log_2(1+\gamma_B)-\log_2(1+\gamma_E),0  \right \} \nonumber \\ 
        &=\left\{ 
        	\begin{array}{ll}
        		\hspace*{1mm} \log_2\left ( \frac{1+\gamma_B}{1+\gamma_E} \right ), \quad \text{if} \enspace \gamma_B>\gamma_E\\
        		\hspace*{1mm} 0, \hspace{6em} \text{if} \enspace \gamma_B \leq \gamma_E,
        	\end{array}
        \right. \vspace{2mm} 
\end{align}
where $P_{\mathrm{A}}$ is the transmit power at Alice, $N_{\mathrm{0}}$ is the average noise power, and  $C_B$ and $C_E$ are the capacities of the main and wiretap
channels, respectively. Hence, the secrecy outage probability (SOP)  is defined as the probability that the instantaneous secrecy capacity falls below a target secrecy rate threshold $R_{th}$~\cite{Wyner}, thus being given by 

 \begin{align}\label{eq:sop}
 \text{SOP}&=\Pr\left \{ C_s\left ( \gamma_B,\gamma_E \right ) < R_{th}  \right \} 
=\Pr\left \{ \left ( \frac{1+\gamma_B}{1+\gamma_E} \right ) < 2^{R_{th}}  \right \}\nonumber \\ 
&\stackrel{(a)}{\geq} \Pr\left \{ \frac{\gamma_B}{\gamma_E}< 2^{R_{th}}\buildrel \Delta \over  = \tau_1 \right \}\nonumber \\ 
 &=F_{X_1}(\tau_1)
\end{align}
where $X_1=\gamma_B/\gamma_E$ and $F_{X_1}(\cdot)$ is a CDF given as in~\eqref{eq:CDFRATIO}. In step $(a)$, we have considered a lower bound of the SOP, which results very tight, as shall be shown in Section \ref{sect:numericals}. 
It is noteworthy that, our formulation for the lower bound of the SOP is valid for non-constrained arbitrary values of the fading parameters corresponding to the main channel and eavesdropper channel (i.e., $\alpha_i$ and $\mu_i$, for $i$ $\in \left \{ B,E \right \}$). This is in contrast to previous works~\cite{Lei,Kong} related to the performance analysis of physical layer security over single-input single-output (SISO) $\alpha$-$\mu$ fading channels, where constraints on the fading parameter values were considered (more specifically, $\alpha_B{=}\alpha_E$ in~\cite{Kong}, and $\alpha_B$, $\alpha_E$  must be co-prime integers in~\cite{Lei}). Therefore, our expressions are a generalization of the aforementioned approaches.

\subsection{Outage Performance of Cognitive Relaying Networks}

Cognitive relaying networks is another application where the statistics of the ratio of  RVs appear. In particular, consider the cognitive relaying network depicted in Fig.~\ref{sistema2}.
\begin{figure}[!b]
\centering 
\psfrag{P}[Bc][Bc][0.8]{P}
\psfrag{S}[Bc][Bc][0.8]{S}
\psfrag{R}[Bc][Bc][0.8]{R}
\psfrag{D}[Bc][Bc][0.8]{D}
\psfrag{V}[Bc][Bc][0.8][30]{$h_{\mathrm{SP}}$}
\psfrag{Z}[Bc][Bc][0.8][0]{$h_{\mathrm{SR}}$}
\psfrag{W}[Bc][Bc][0.8][0]{$h_{\mathrm{RD}}$}
\psfrag{U}[Bc][Bc][0.8][40]{$h_{\mathrm{RP}}$}
\includegraphics[width=0.7\linewidth]{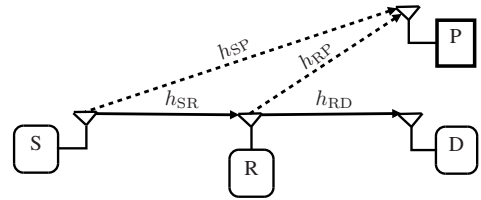} \caption{System model of an underlay cognitive relaying network. The data links are represent by solid lines, while the
interference links are represented by dashed lines. }
\label{sistema2}
\end{figure}
In this system, a secondary network consisting of one secondary source (S), one secondary decode-and-forward (DF) relay (R), and one secondary destination (D), operate by sharing the spectrum belonging to a primary network. Thus, the secondary transmissions are subject to power constraints inflicted by a primary destination (P) in an underlay spectrum-sharing scenario, so that a predetermined level of interference temperature on the primary receiver is satisfied~\cite{art:haykin}. In this system, the direct link is neglected, as it is considered to be extremely attenuated, and all terminals are assumed to be equipped with a single antenna. The channel coefficients of the data links ${\mathrm{S}}\rightarrow {\mathrm{R}}$ and ${\mathrm{R}}\rightarrow {\mathrm{D}}$ are denoted by $h_{\mathrm{SR}}$ and $h_{\mathrm{RD}}$, respectively; and the channel coefficients of the interference links ${\mathrm{S}}\rightarrow {\mathrm{P}}$ and ${\mathrm{R}}\rightarrow {\mathrm{P}}$ are denoted by $h_{\mathrm{SP}}$ and $h_{\mathrm{RP}}$, respectively. Thus, the corresponding channel power gains $g_{i,j}=\abs{h_{i,j}}^{2},$ with $i \in \left \{ {\mathrm{R}}, {\mathrm{S}} \right \}$ and $j\in \left \{ {\mathrm{D}}, {\mathrm{P}}, {\mathrm{R}} \right \} $,  are subject to block $\alpha$-$\mu$ fading. The maximum interference power tolerated at ${\mathrm{P}}$, coming from the cognitive network, is denoted by $I$. It is assumed that the transmit powers at the secondary source and relay are  $P_{\mathrm{S}}{=}I/g_{\mathrm{SP}}$ and $P_R{=}I/g_{\mathrm{RP}}$, respectively.
In addition, $\overline{\gamma}_I \buildrel \Delta \over = I/N_{\mathrm{0}}$ is defined as the maximum interference-to-noise ratio tolerated
at the primary destination.
Then, the instantaneous received signal-to-noise ratio (SNR) at the secondary relay and the secondary destination are given, respectively, by
\begin{align}
\gamma_{\mathrm{SR}}&=\frac{g_{\mathrm{SR}} P_{\mathrm{S}}}{N_{\mathrm{0}}}=\frac{g_{\mathrm{SR}} I}{g_{\mathrm{SP}} N_{\mathrm{0}}}=\frac{g_{\mathrm{SR}} \overline{\gamma}_I}{g_{\mathrm{SP}} },\\
\gamma_{\mathrm{RD}} & =\frac{g_{\mathrm{RD}} P_{\mathrm{R}}}{N_{\mathrm{0}}}=\frac{g_{\mathrm{RD}} I}{g_{\mathrm{RP}} N_{\mathrm{0}}}=\frac{g_{\mathrm{RD}} \overline{\gamma}_I}{g_{\mathrm{RP}}}.
\end{align}
The outage probability of the secondary network for the DF relaying protocol can be written as~\cite{edgar}
\begin{align}
\nonumber P_{\mathrm{out}}=&\Pr \left(\min\bigg\{\gamma_{\mathrm{SR}},\gamma_{\mathrm{RD}}\bigg\}<2^{2\mathcal{R}}-1 \buildrel \Delta \over  = \tau_2\right)\\
=& F_{X_2}\left(\tau_2\right) +F_{X_3}\left(\tau_2\right)- F_{X_2}\left(\tau_2\right) F_{X_3}\left(\tau_2\right),
\end{align}
where $F_{X_2}(\cdot)$ and $F_{X_3}(\cdot)$ are the CDFs for the RVs $X_2{=}\gamma_{\mathrm{SR}}$ and $X_3{=}\gamma_{\mathrm{RD}}$, respectively, which can be evaluated as in~\eqref{eq:CDFRATIO}, $\mathcal{R}$ is the target rate and $\tau_2$ is the target SNR threshold.

\subsection{Outage Performance of Full-Duplex Relaying Networks}
Another application where the statistics of the ratio of independent squared $\alpha$-$\mu$ random variables are considered is in FD relaying systems~\cite{art:osorio,art:olivo}. Let us consider the system depicted in Fig.~\ref{sistema3}, which illustrates a two-hop FD relaying network composed of three nodes: one single-antenna source (S), one single-antenna destination (D), and one DF relay (R) equipped with one transmit antenna and one receive antenna to operate in full-duplex mode.
\begin{figure}[!t]
\centering 
\psfrag{S}[Bc][Bc][0.8]{S}
\psfrag{R}[Bc][Bc][0.8]{R}
\psfrag{D}[Bc][Bc][0.8]{D}
\psfrag{X}[Bc][Bc][0.8][0]{$h_{\mathrm{SR}}$}
\psfrag{Y}[Bc][Bc][0.8][0]{$h_{\mathrm{RD}}$}
\psfrag{U}[Bc][Bc][0.8]{$h_{\mathrm{RR}}$}
\includegraphics[width=0.7\linewidth]{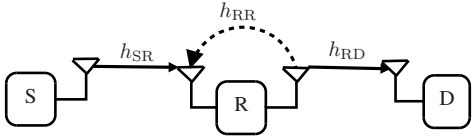} \caption{System model of a three-node FD relaying network
(data link: solid line; interference link: dashed line).
}
\label{sistema3}
\end{figure}
In this system, it is assumed that the direct link is highly attenuated, thus being neglected. Moreover, all channels in this network are subject to block $\alpha$-$\mu$ fading. Thus, $\gamma_{\mathrm{SR}}=|h_{\mathrm{SR}}|^2\gamma_P/2$ and $\gamma_{\mathrm{RD}}=|h_{\mathrm{RD}}|^2\gamma_P/2$ are the instantaneous received SNRs for the first- and second-hop relaying links, respectively, where $h_{\mathrm{SR}}$ and $h_{\mathrm{RD}}$ are the corresponding channel coefficients, and $\gamma_P$ is the transmit system SNR. Moreover, due to imperfect stages of interference cancellation at the FD relay, a residual self-interference (RSI) is considered, which can be modeled as a Rayleigh fading loop back channel~\cite{art:olivo,art:osorio}, with channel coefficient $h_{\mathrm{RR}}{\sim}\mathcal{CN}\left(0,\sigma^2\right)$, such that $\gamma_{\mathrm{RR}}=|h_{\mathrm{RR}}|^2\gamma_P/2$ is the instantaneous received SNR.

Considering the DF relaying protocol, the outage probability of the system under study can be formulated as~\cite{art:olivo}
\begin{align}
\nonumber P_{\mathrm{out}} = &\Pr \left(\min\bigg\{\dfrac{\gamma_{\mathrm{SR}}}{\gamma_{\mathrm{RR}}+1},\gamma_{\mathrm{RD}}\bigg\}<2^{\mathcal{R}}-1 \buildrel \Delta \over  = \tau_3\right)\\
\approx & F_{X_4}\left(\tau_3\right) +F_{X_5}\left(\tau_3\right)- F_{X_4}\left(\tau_3\right) F_{X_5}\left(\tau_3\right),
\end{align}
where, by considering an interference-limited scenario, such that $\gamma_{\mathrm{SR}}/(\gamma_{\mathrm{RR}}+1)\approx \gamma_{\mathrm{SR}}/\gamma_{\mathrm{RR}}$, $F_{X_4}(\cdot)$ is the CDF of the RV $X_4={\gamma_{\mathrm{SR}}}/{\gamma_{\mathrm{RR}}}$ and $F_{X_5}(\cdot)$ is the CDF of the RV $\gamma_{\mathrm{RD}}$, both of which being straightforwardly evaluated as in~\eqref{eq:CDFRATIO}.

\section{Numerical results and discussions} \label{sect:numericals}
In this section, we validate the accuracy of the proposed expressions for 
some representative cases via Monte Carlo simulations. 

Figs.~\ref{PDFV2} and~\ref{PCFV2} respectively show the PDF and CDF obtained for the ratio of two squared $\alpha$-$\mu$ RVs, by considering different values of fading parameters. In both figures, the values of the fading parameters are chosen to show the wide range of shapes that the distribution of the ratio can assume. Fig.~\ref{PDFV2} illustrates the resulting PDF for different values of $\left \{\mu_1, \mu_2  \right \}$, with $\left \{\alpha_1, \alpha_2  \right \}=\left \{1.5, 1.1  \right \}$ and $\bar{\Upsilon}_1=\bar{\Upsilon}_2= 0$ dB. It can be observed that our expressions perfectly match the Monte Carlo simulations, thus validating our results. 

Fig.~\ref{PCFV2} shows the resulting CDF for distinct values of $\left \{\alpha_1, \alpha_2  \right \}$, with $\left \{\mu_1, \mu_2  \right \}=\left \{3.5, 2.8  \right \}$ and $\bar{\Upsilon}_1=\bar{\Upsilon}_2= 0$~dB. Once again, it is observed that our expressions perfectly match the Monte Carlo simulations. It can also be noticed from the cases presented in those figures that our expressions allow non-constrained arbitrary values of fading. 
\begin{figure}[H]
\centering
\includegraphics[width=0.9\columnwidth]{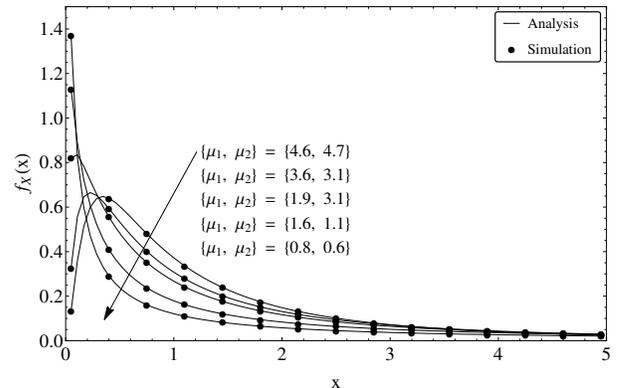}
    \vspace{-4mm}
	\caption{PDF of the ratio of two squared $\alpha$-$\mu$  RVs for different values of $\left \{ \mu_1,
\mu_2\right \}$, with $\left \{\alpha_1, \alpha_2  \right \} = \left \{ 1.5,1.1 \right \} $ and $\bar{\Upsilon}_1=\bar{\Upsilon}_2= 0$ dB. 
}
\label{PDFV2}
\end{figure}
\begin{figure}[H]
\centering
\includegraphics[width=0.9\columnwidth]{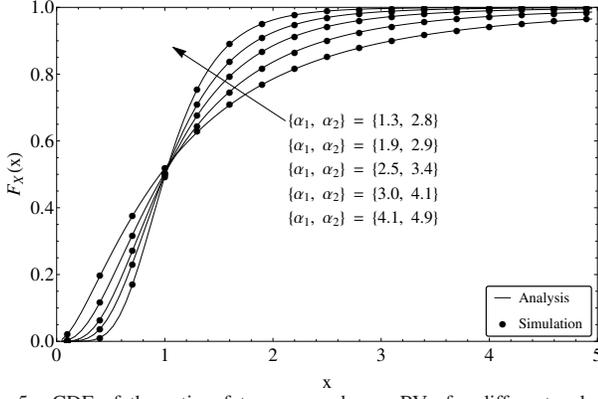}
    \vspace{-4mm}
	\caption{CDF of the ratio of two squared $\alpha$-$\mu$  RVs for different values of $\left \{  \alpha_1,
\alpha_2\right \}$, with $\left \{\mu_1, \mu_2  \right \} = \left \{ 3.5, 2.8 \right \} $ and $\bar{\Upsilon}_1=\bar{\Upsilon}_2= 0$ dB.}
\label{PCFV2}
\end{figure}	

\begin{figure}[H]
\centering
\includegraphics[width=0.9\columnwidth]{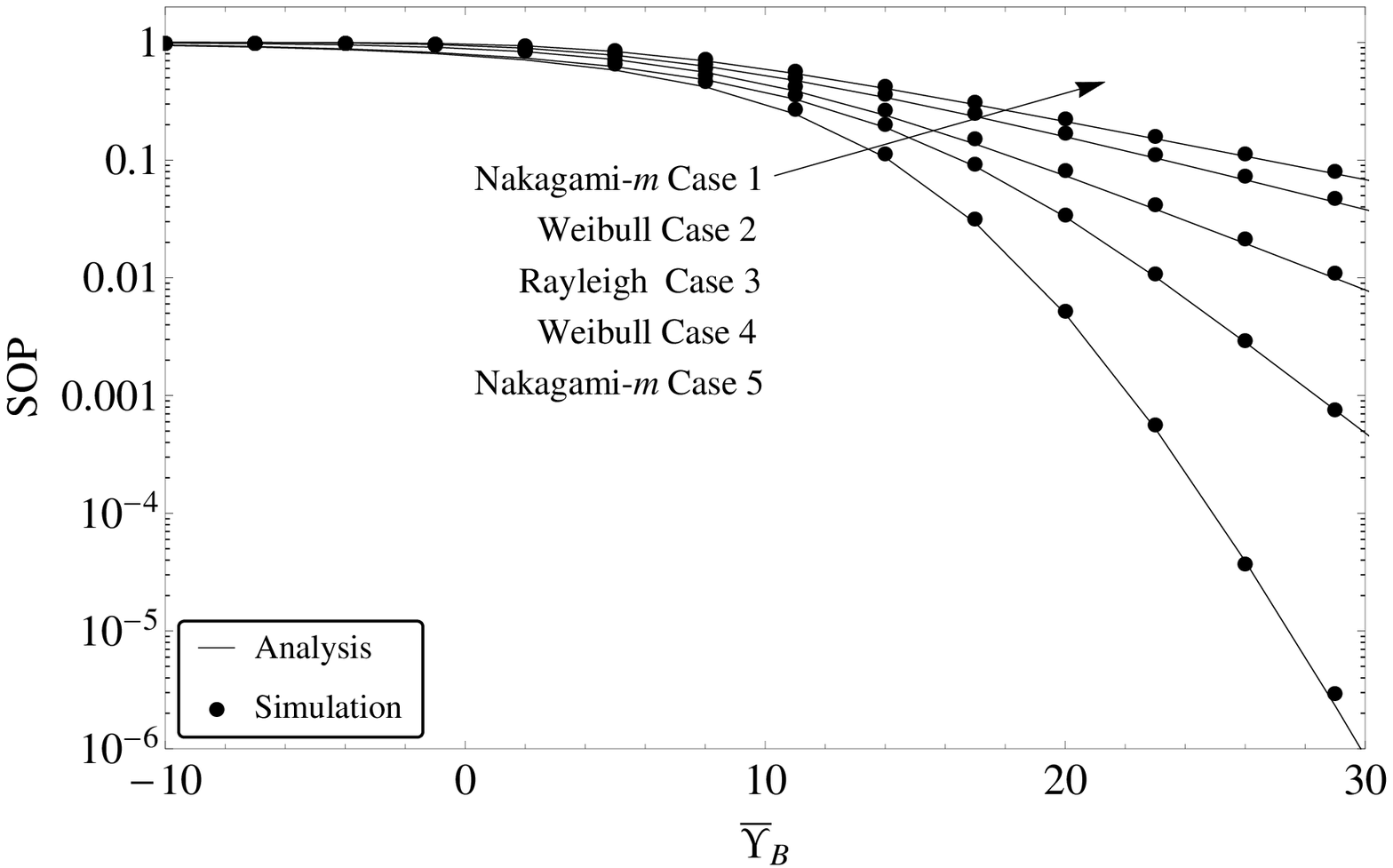}
    \vspace{-4mm}
	\caption{Secrecy outage probability versus $\bar{\Upsilon}_\mathrm{B}$ for different combinations of $\left \{  \alpha_\mathrm{B}, \mu_\mathrm{B}, \alpha_\mathrm{E}, \mu_\mathrm{E}  \right \} $, and $\bar{\Upsilon}_\mathrm{E}=1$ dB and $\tau_1= 1$.}
\label{SOPV2}
\end{figure}
Fig.~\ref{SOPV2} shows the SOP versus $\bar{\Upsilon}_\mathrm{B}$ for different combinations of fading parameters, with $\bar{\Upsilon}_\mathrm{B}=1$~dB and $\tau_1= 0$~dB. More specifically, we set the fading parameters~to the following cases: 
\begin{itemize}
\item \textit{Case 1:} Nakagami-$m$\\$\left \{ \alpha_\mathrm{B}, \mu_\mathrm{B}\right \}=\left \{2, 4.5\right \}$,  $\left \{ \alpha_\mathrm{E},\mu_\mathrm{E} \right \}=\left \{2, 0.6\right \}$.

\item \textit{Case 2:} Weibull \\ $\left \{\alpha_\mathrm{B},\mu_\mathrm{B}\right \}=\left \{3.9, 1\right \}$, $\left \{\alpha_\mathrm{E},\mu_\mathrm{E}\right \}=\left \{1.3, 1\right \}$.

\item \textit{Case 3:} Rayleigh \\ $\left \{\alpha_\mathrm{B},\mu_\mathrm{B}\right \}=\left \{2, 1\right \}$, $\left \{\alpha_\mathrm{E},\mu_\mathrm{E}\right \}=\left \{2, 1\right \}$.

\item \textit{Case 4:} Weibull\\ $\left \{  \alpha_\mathrm{B},\mu_\mathrm{B}\right \}=\left \{ 1.2, 1\right \}$, $\left \{ \alpha_\mathrm{E},\mu_\mathrm{E}\right \}=\left \{4.5, 1\right \}$.

\item \textit{Case 5:}Nakagami-$m$ \\ $\left \{\alpha_\mathrm{B},\mu_\mathrm{B}\right \}=\left \{  2, 0.5\right \}$, $\left \{\alpha_\mathrm{E},\mu_\mathrm{E}\right \}=\left \{ 2, 3.1\right \}$.
\end{itemize}
For all cases, it can be noticed that the proposed lower bound is very tight to the exact SOP obtained by Monte Carlo simulations. Also, it is observed that, in general, the secrecy performance worsens as $\alpha_\mathrm{B}$, $\mu_\mathrm{B}$ decrease and $\alpha_\mathrm{E}$, $\mu_\mathrm{E}$ increase (see, e.g., cases~2, 3, and 4), which are the fading parameters of the main and eavesdropper channels, respectively. In contrast, 
note that the secrecy performance improves as $\alpha_\mathrm{B}$, $\mu_\mathrm{B}$ increase and $\alpha_\mathrm{E}$, $\mu_\mathrm{E}$ decrease (i.e., the eavesdropper channel is in a worse channel condition). Importantly, this fact implies that the fading conditions can be exploited to prevent the information from 
being overheard by an eavesdropper.


\begin{figure}[H]
\centering
\includegraphics[width=0.9\columnwidth]{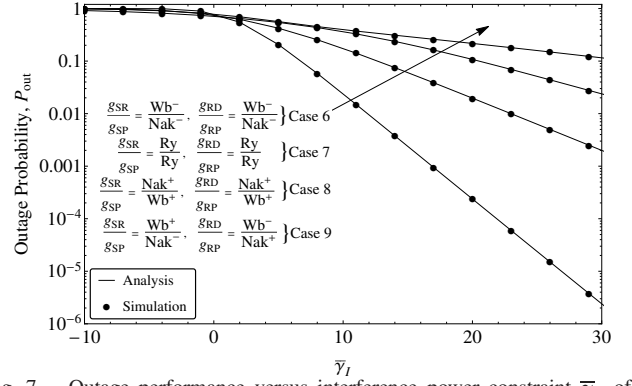}
    \vspace{-4mm}
	\caption{ Outage performance versus interference power constraint $\overline{\gamma}_I$ of a cognitive relaying network, for different values of fading parameters. Notation: $\mathrm{Ry}\rightarrow $ Rayleigh, $\mathrm{Wb}^+\rightarrow $ severe Weibull, $\mathrm{Wb}^-\rightarrow $ weak Weibull, $\mathrm{Nak}^+\rightarrow $ severe Nakagami-$m$, $\mathrm{Nak}^-\rightarrow $ weak Nakagami-$m$.    	 }
\label{CRNV2}
\end{figure}

Fig.~\ref{CRNV2} illustrates the influence of the fading parameters on the outage performance of a cognitive relaying network. This figure shows the outage probability versus the maximum interference power constraint at the primary receiver, $\overline{\gamma}_I$, for $\bar{\Upsilon}_{\mathrm{SP}}=\bar{\Upsilon}_{\mathrm{SR}}=\bar{\Upsilon}_{\mathrm{RP}}=\bar{\Upsilon}_{\mathrm{RD}}= 1$ dB and a target SNR threshold $\tau_2=0$~dB.
For these scenarios, the fading parameters are set to the next cases:
\begin{itemize}
\item \textit{Case 6:} $\left \{ \alpha_\mathrm{SR}, \mu_\mathrm{SR}\right \}=\left \{4.2, 1\right \}$,$\left \{ \alpha_\mathrm{SP},\mu_\mathrm{SP} \right \}=\left \{2, 4.1\right \}$, $\left \{ \alpha_\mathrm{RD},\mu_\mathrm{RD} \right \}=\left \{3.9, 1\right \}$, $\left \{ \alpha_\mathrm{RP},\mu_\mathrm{RP} \right \}=\left \{2, 3.8\right \}$.
\item  \textit{Case 7:} $\left \{ \alpha_\mathrm{SR},\mu_\mathrm{SR}\right \}=\left \{2, 1\right \}$, $\left \{ \alpha_\mathrm{SP},\mu_\mathrm{SP}\right \}=\left \{2, 1\right \}$, $\left \{ \alpha_\mathrm{RD},\mu_\mathrm{RD} \right \}=\left \{2, 1 \right \}$, $\left \{ \alpha_\mathrm{RP},\mu_\mathrm{RP} \right \}=\left \{ 2, 1\right \}$.
\item \textit{Case 8:} $\left \{ \alpha_\mathrm{SR},\mu_\mathrm{SR}\right \}=\left \{2, 0.6\right \}$, $\left \{ \alpha_\mathrm{SP},\mu_\mathrm{SP}\right \}=\left \{ 0.8, 1\right \}$, $\left \{ \alpha_\mathrm{RD},\mu_\mathrm{RD} \right \}=\left \{2, 0.9\right \}$, $\left \{ \alpha_\mathrm{RP},\mu_\mathrm{RP} \right \}=\left \{0.7, 1\right \}$.
\item \textit{Case 9:} $\left \{ \alpha_\mathrm{SR},\mu_\mathrm{SR}\right \}=\left \{0.6, 1\right \}$, $\left \{ \alpha_\mathrm{SP},\mu_\mathrm{SP}\right \}=\left \{2, 4.2\right \}$, $\left \{ \alpha_\mathrm{RD},\mu_\mathrm{RD} \right \}=\left \{4.1, 1\right \}$, $\left \{ \alpha_\mathrm{RP},\mu_\mathrm{RP} \right \}=\left \{2, 0.8\right \}$.
\end{itemize}
It can be observed from all the curves that our analytical expression matches the Monte Carlo simulations. Moreover, note that, as the fading parameters increase, i.e., for better channel conditions (see, e.g., Case~6, and Case~7), the outage performance improves, as expected. In the opposite scenario, i.e., signals with lower values of the fading parameters (see, e.g., Case~8, and Case~9), the outage performance worsens, due to poor channel conditions. In addition,
the outage behavior improves as $\overline{\gamma}_I$ increases, as expected.


\begin{figure}[H]
\centering
\includegraphics[width=0.9\columnwidth]{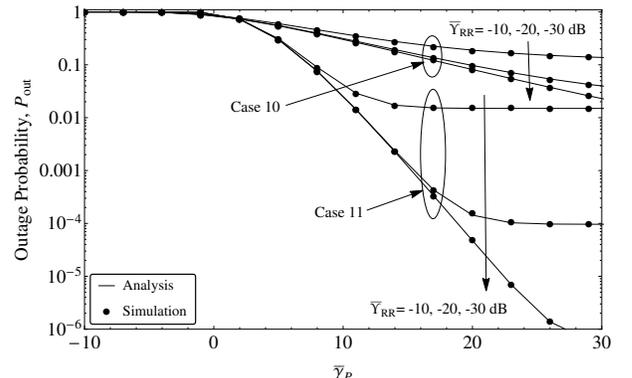}
    \vspace{-4mm}
	\caption{ Outage performance versus $\overline{\gamma}_P$ of a FD relaying network, by considering distinct values of average channel power gain at the RSI link, $\bar{\Upsilon}_{RR}$. Two cases are considered: severe fading (case 10) and weak fading (case 11). }
\label{FDV2}
\end{figure}

Fig.~\ref{FDV2} shows the outage performance of a FD relaying network versus the transmit system SNR for $\bar{\Upsilon}_{SR}=\bar{\Upsilon}_{RD}=0$ dB, $\tau_3= 0$~dB, and different values of average channel power gain at the RSI link, namely, $\bar{\Upsilon}_{RR}=-10, -20, -30$ dB. For these scenarios, the values of the fading parameters are set to teh next cases:
\begin{itemize}
\item \textit{Case~10:} Severe fading \\ $\left \{ \alpha_{\mathrm{SR}},\mu_{\mathrm{SR}}\right \}=\left \{1.8, 0.8\right \}$, $\left \{ \alpha_{\mathrm{RR}},\mu_{\mathrm{RR}} \right \}=\left \{2.2, 0.7\right \}$, $\left \{ \alpha_{\mathrm{RD}},\mu_{\mathrm{RD}} \right \}=\left \{2.1, 0.6\right \}$.
\item \textit{Case~11:} Weak fading \\ $\left \{ \alpha_{\mathrm{SR}},\mu_{\mathrm{SR}}\right \}=\left \{1.9, 2.3\right \}$, $\left \{ \alpha_{\mathrm{RR}},\mu_{\mathrm{RR}} \right \}=\left \{2.1, 2.8\right \}$, $\left \{ \alpha_{\mathrm{RD}},\mu_{\mathrm{RD}} \right \}=\left \{2.2, 2.9\right \}$.
\end{itemize}
In a similar manner, note that our analytical results are highly accurate with respect to the Monte Carlo simulations, thus confirming the correctness of our derivations. We can also observe that the average channel power gain of the RSI link affects the system performance in a different manner according to the fading parameters of the channel. For instance, when dealing with weak fading (e.g., Case~11), the outage performance shows significant improvements as the level of RSI decreases. On the other hand, for a severe fading case (e.g., Case~10), it is observed that, even though for lower values of the level of RSI, the improvement on the performance is not significant. Also, it is observed a performance floor in the medium-to-high SNR regime. This behavior is caused by the RSI at the FD relay. In this context, it is worth mentioning that self-interference mitigation techniques play a pivotal role in exploiting the potential benefits of FD relaying, mainly at the medium-to-high SNR region. Additionally, it can be noticed that the consideration of a more general distribution in the case of FD, lead us to a more comprehensive analysis of different scenarios according to the severity of fading. Also, interested readers can revise~\cite{FD1,FD2} for further guidance about self-interference cancellation on FD relay systems.


\section{Conclusions}
In this paper, novel exact analytical expressions for
the PDF, CDF, MGF, and higher order moments of the ratio of two squared $\alpha$-$\mu$ RVs in terms of the Fox H-function were derived. Importantly, this expressions, unlike previous related works, are valid for any values of the fading parameters $\alpha$ and $\mu$. Additionally, a series representation for the formulations are also provided. Based on these results, analytical expressions for the statistics of the ratio of well-known distributions, such as Nakagami-$m$, Weibull, and Rayleigh, were also provided as byproducts. These novel statistics represent a useful tool to assess the performance of wireless communication schemes considering generalized fading-channel models with applicability in scenarios for next-generation wireless networks. For illustration purposes, we analyze three application uses by analyzing ($i$) the secrecy outage probability for PLS-based wireless networks, ($ii$) the outage performance for cognitive relaying networks, and ($iii$) the outage performance for FD relaying networks. The obtained analytical expressions were validated by Monte Carlo simulations. Finally, it is worthwhile to mention that the analytical
results presented in this work can be evaluated in a straightforward and efficient manner through mathematical software packages. For this purpose, we have also provided an implementation of the Fox H-function.

\appendices

\section{ Proof of Proposition~\ref{prop:pdf}}
\label{ap:statistics}
Assuming that $\Upsilon_1$ and $\Upsilon_2$ are statistically independent, the PDF of $X$ can be obtained as~\cite{Leonardo}
\begin{align}\label{eq:aneA1}
 f_X(x)&=\int_{0}^{\infty}y f_{\Upsilon_1}\left (x y \right )f_{\Upsilon_2}(y)dy.
 \end{align}
 Now, by substituting~\eqref{eq:6} into~\eqref{eq:aneA1}, it follows that
 \begin{align}\label{eq:aneA2}
 f_X(x)&=\frac{\alpha_1\alpha_2 x^{\frac{\alpha_1\mu_1}{2}-1}}{4\beta_2^{\frac{\alpha_2\mu_2}{2}} \beta_1^{\frac{\alpha_1\mu_1}{2}}\Gamma (\mu_2)\Gamma (\mu_1)} \nonumber \\  & \times
 \int_{0}^{\infty}y^{\frac{\alpha_2\mu_2}{2}+\frac{\alpha_1\mu_1}{2}-1}G_{0,1}^{1,0}\left[ \left ( \frac{x y}{\beta_1}  \right )^{\frac{\alpha_1}{2}} \bigg| 
\begin{array}{c}
 0\\
\end{array}
\right] \nonumber \\ &\times
G_{0,1}^{1,0}\left[ \left ( \frac{y}{\beta_2}  \right )^{\frac{\alpha_2}{2}} \bigg|
\begin{array}{c}
 0\\
\end{array}
\right]dy.
\end{align}
After some mathematical manipulations in~\eqref{eq:aneA2}, we have that
 \begin{align}\label{eq:aneA3}
 f_X(x)&=\frac{\alpha_1 x^{\frac{\alpha_1\mu_1}{2}-1}}{2\beta_2^{\frac{\alpha_2\mu_2}{2}} \beta_1^{\frac{\alpha_1\mu_1}{2}}\Gamma (\mu_2)\Gamma (\mu_1)}  \underset{I_1}{\underbrace{\int_{0}^{\infty}w^{\mu_2+k\mu_1-1} }}\nonumber \\  &
 \underset{I_1}{\underbrace{\times G_{0,1}^{1,0}\left[\frac{w}{\beta_2^{\frac{\alpha_2}{2}}}   \bigg|
\begin{array}{c}
 0\\
\end{array}
\right]G_{0,1}^{1,0}\left[ \frac{w^{k}}{\left (  \frac{x}{\beta_1}  \right )^{\frac{-\alpha_1}{2}} }\bigg|
\begin{array}{c}
 0\\
\end{array}
\right]dw,}}
\end{align}
where $w{=}y^{\alpha_2/2}$ and recalling that $k{=}\frac{\alpha_1}{\alpha_2}$. Then, by using~\cite[Eq. (07.34.21.0009.01)]{Wolfram1}, $I_1$ in~\eqref{eq:aneA3} can be solved in  a straightforward manner as
\begin{align}\label{eq:aneA4}
I_1&=\left (\frac{1}{\beta_2^{\frac{\alpha_2}{2}}}\right ) ^{-(\mu_2+k\mu_1)}\nonumber \\
&\times \mathrm{H}_{1,1}^{1,1}\left[\left ( \frac{x\beta_2}{\beta_1}  \right )^{\frac{\alpha_2}{2}}\bigg|
\begin{array}{c}
 (1-\mu_2-k\mu_1,k) \\
 (0,1)\\
\end{array}
\right].
\end{align}
Finally, by replacing $I_1$ into~\eqref{eq:aneA3}, we obtain the expression in~\eqref{pdfRatio}.

 
On the other hand, the CDF of $X =\Upsilon_1/\Upsilon_2 $ can be formulated
as
\begin{align}\label{eq:cdfratios}
 F_X(x)&=\Pr\left \{ X \leq x  \right \}\nonumber \\ 
 &=\Pr\left \{\frac{\Upsilon_1}{\Upsilon_2}\leq x  \right \}\nonumber \\
  &=\Pr\left \{\Upsilon_1\leq x\Upsilon_2  \right \}\nonumber \\
 &= \int_{0}^{\infty}F_{\Upsilon_1}\left (xy \right )f_{\Upsilon_2}(y)dy.
 \end{align}
Then, by replacing~\eqref{eq:6} and~\eqref{eq:7} into~\eqref{eq:cdfratios}, we get
 \begin{align}\label{eq:30}
 F_X(x)&=\frac{\alpha_2x^{\frac{\mu_1\alpha_1}{2}}}{2\beta_2^{\frac{\alpha_2\mu_2}{2}}\beta_1^{\frac{\alpha_1\mu_1}{2}}\Gamma (\mu_2)\Gamma (\mu_1)} \nonumber \\  & \times
 \int_{0}^{\infty}y^{\frac{\alpha_2\mu_2}{2}+\frac{\alpha_1\mu_1}{2}-1}G_{0,1}^{1,0}\left[ \left ( \frac{y}{\beta_2}  \right )^{\frac{\alpha_2}{2}} \bigg|
\begin{array}{c}
 0\\
\end{array}
\right] \nonumber \\ &\times
G_{1,2}^{1,1}\left[ \left ( \frac{x y}{\beta_1}  \right ) ^{\frac{\alpha_1}{2}} \bigg|
\begin{array}{c}
 1-\mu_1 \\
 0,-\mu_1 \\
\end{array}
\right]dy.
\end{align}
Here we proceed by following a similar procedure as in the derivation of the PDF of $X$. By replacing $w=y^{\alpha_Y/2}$ and $k=\frac{\alpha_1}{\alpha_2}$ into~\eqref{eq:30}, it follows that 
\begin{align}\label{eq:31}
 F_X(x)&=\frac{x^{\frac{\mu_1\alpha_1}{2}}}{\beta_2^{\frac{\alpha_2\mu_2}{2}}\beta_1^{\frac{\alpha_1\mu_1}{2}}\Gamma (\mu_2)\Gamma (\mu_1)} \underset{I_2}{\underbrace{\int_{0}^{\infty}w^{\mu_1k+\mu_2-1}}}   \nonumber \\ &\underset{I_2}{\underbrace{ \times G_{0,1}^{1,0}\left[  \frac{w}{\beta_Y^{\frac{\alpha_Y}{2}}} \bigg| 
\begin{array}{c}
 0\\
\end{array}
\right] 
G_{1,2}^{1,1}\left[ \left ( \frac{x}{\beta_1}  \right ) ^{\frac{\alpha_1}{2}} w^k\bigg|
\begin{array}{c}
 1-\mu_1 \\
 0,-\mu_X \\
\end{array}
\right]dw.}}
\end{align}
Now, by using~\cite[Eq. (07.34.21.0009.01
)]{Wolfram1}, $I_2
$ in~\eqref{eq:31} can be solved in  a straightforward manner as
\begin{align}\label{eq:33}
I_2&=\left (\frac{1}{\beta_2^{\frac{\alpha_2}{2}}}\right ) ^{-\left (\mu_1k+\mu_2\right )}\nonumber \\
&\times \mathrm{H}_{2,2}^{1,2}\left[\left ( \frac{x\beta_2}{\beta_1}  \right )^{\frac{\alpha_1}{2}}\bigg|
\begin{array}{c}
 (1-\mu_1,1),(1-\mu_1k-\mu_2,k) \\
 (0,1) ,\hspace{0.5mm} (-\mu_1,1)\\
\end{array}
\right].
\end{align}
Finally, substituting~\eqref{eq:33} into~\eqref{eq:31}, a closed-form
expression for the CDF of $X{=}\Upsilon_1/\Upsilon_2$ can be calculated as in~\eqref{eq:CDFRATIO}.

Now, the MGF of $X=\Upsilon_1/\Upsilon_2$ can be obtained, by definition, as~\cite{mgf} 
\begin{align}\label{eq:mgf1}
{\cal M}_\text{X}(s)& \buildrel \Delta \over = \mathbb{E}\left[ \text{e}^{-s X} \right ]=\int_{0}^{\infty }\exp\left (-sx \right )f_X(x)dx.\nonumber \\
\end{align}
By replacing $f_X(x)$ given as in~\eqref{pdfRatio} into~\eqref{eq:mgf1}, we obtain
\begin{align}\label{mgf2}
{\cal M}_\text{X}(s)&=\frac{\alpha_{1} }{2 \Gamma (\mu_2)\Gamma (\mu_1)}\left ( \frac{\beta_2}{\beta_1} \right )^{\frac{\alpha_1\mu_1}{2}} \int_{0}^{\infty }x^{\frac{\alpha_1\mu_1}{2}-1} \nonumber \\
&\times e^{-s x}
H_{1,1}^{1,1}\left[\left ( \frac{x\beta_2}{\beta_1}  \right )^{\frac{\alpha_1}{2}}\bigg|
\begin{array}{c}
 (1-\mu_2-k\mu_1,k)\\
 (0,1)\\
\end{array}
\right]dx.
\end{align}
Substituting the Fox H-function in~\eqref{mgf2} by its Mellin-Barnes type contour integral as in~\cite[Eq. (1.2)]{Fox}, interchanging the order of integrations, and performing some simplifications, we obtain 
\begin{align}\label{mgf3}
{\cal M}_\text{X}(s)&=\frac{\alpha_{1} \left ( \frac{\beta_2}{\beta_1} \right )^{\frac{\alpha_1\mu_1}{2}} }{2 \Gamma (\mu_2)\Gamma (\mu_1)}\int_{0}^{\infty }x^{\frac{\alpha_1\mu_1}{2}-1} \exp\left (-sx \right )  \nonumber \\
&
\times \frac{1}{2\pi \mathrm{i}}\int_{\mathcal{C}}^{ }\Gamma(z)\Gamma(\mu_2+k\mu_1-kz)\left [ \left ( \frac{x\beta_2}{\beta_1} \right )^{\frac{\alpha_1}{2}}  \right ]^{-z}dzdx \nonumber \\
&= \frac{\alpha_{1} \left ( \frac{\beta_2}{\beta_1} \right )^{\frac{\alpha_1\mu_1}{2}} }{4 \Gamma (\mu_2)\Gamma (\mu_1)\mathrm{i}}\int_{\mathcal{C}}^{ }\Gamma(z)\Gamma(\mu_2+k\mu_1-kz)\nonumber \\
&\times  \left [ \left ( \frac{\beta_2}{\beta_1} \right )^{\frac{\alpha_1}{2}}  \right ]^{-z}\int_{0}^{\infty }x^{\frac{\alpha_1\mu_1}{2}-\frac{z\alpha_1}{2}-1} \exp\left (-sx \right ) dxdz\nonumber \\
&=\frac{\alpha_{1} \left ( \frac{\beta_2}{\beta_1} \right )^{\frac{\alpha_1\mu_1}{2}} s^{-\frac{\alpha_1 \mu_1}{2}}  }{2 \Gamma (\mu_2)\Gamma (\mu_1)} \underset{I_3}{\underbrace{ \frac{1}{2\pi \mathrm{i}}\int_{\mathcal{C}}^{ }\Gamma(z)\Gamma(\mu_2+k\mu_1-kz) }}\nonumber \\
&\underset{I_3}{\underbrace{\times\Gamma \left ( \frac{\mu_1 \alpha_1}{2}-\frac{\alpha_1 z}{2} \right ) \left [ \left ( \frac{\beta_2}{s\beta_1} \right )^{\frac{\alpha_1}{2}}  \right ]^{-z}dz. }}
\end{align}
Then, by substituting $I_3$ in~\eqref{mgf3} by its corresponding Fox H-function with the use of~\cite[Eq. (1.1)]{Fox}, we obtain the expression in~\eqref{eq:MGF}, thus accomplishing the proof.

\section{ }\label{ap:mathimplementation}
\begin{table}[H]
\caption{MATHEMATICA\textregistered IMPLEMENTATION OF THE FOX-H FUNCTION}
\vspace{-2mm}
\centering	\includegraphics[width=\columnwidth]{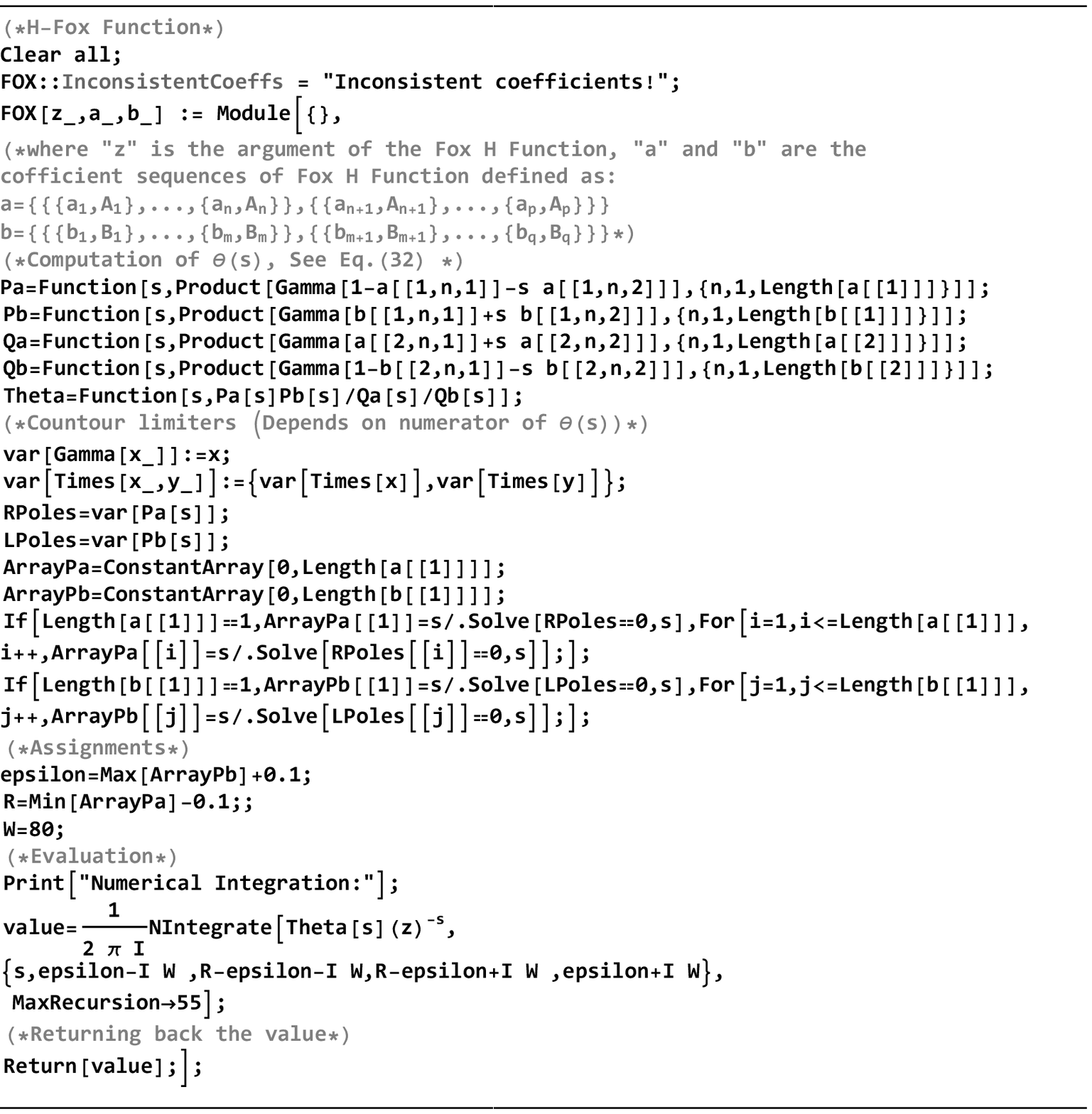}
 \label{Figura1}
\end{table}

\section{ }\label{ap:residues}
Here, for illustration purposes, the Fox H-function in~\eqref{eq:CDFRATIO} is expressed as a sum of residues
~\cite{Carlos}. To this end, we start by defining the Fox H-function as~\cite[Eq.~(1.1)]{Fox}
\begin{align}\label{eq:12}
\mathrm{H}_{p,q}^{m,n}\left [ z \right ]&=\mathrm{H}_{p,q}^{m,n}\left[z \bigg|
\begin{array}{c}
 (a_1,A_1),\dots, (a_p,A_p) \\
 (b_1,B_1),\dots, (b_q,B_q) \\
\end{array}
\right]\nonumber \\
&= \frac{1}{2\pi \mathrm{i}}\int_{\mathcal{C}}^{ }\Theta(s)z^{-s}ds,
\end{align}
where $m$, $n$, $p$, $q$  $\in \field{N}^0$, with $0\leq n\leq p$, $1\leq m\leq q$, $z\in\mathbb{C}\backslash\{0\}$. Here
\begin{multline}\label{eq:13}
\Theta(s)=\frac{\left \{ \prod_{j=1}^{m}\Gamma\left ( b_j+B_js \right )  \right \}}{\left \{\prod_{j=m+1}^{q}\Gamma\left (1-b_j-B_js \right )   \right \}}\\
\times \frac{\left \{ \prod_{j=1}^{n}\Gamma\left (1-a_j-A_js \right )  \right \}}{\left \{\prod_{j=n+1}^{p}\Gamma\left (a_j+A_js \right )   \right \}}.
\end{multline}
An empty product is always interpreted as unity, $A_i, B_j \in \mathbb{R}^+$, $a_i, b_j \in \mathbb{C}$, $i=1,\dots,p$;  $j=1,\dots,q$. In addition, $\mathcal{C}=\left (c-i\infty, c+i\infty  \right )$ is a  contour of integration separating the poles of $\Gamma(1-a_j-A_js)$, $j=1,\cdots,n$ from those of $\Gamma(b_j+B_js)$, $j=1,\cdots,m$. On the other hand, the contour integral $\mathcal{C}$ in~\eqref{eq:12} can be obtained by the sum of residues technique, evaluated at all poles of $\Theta(s)$~\cite{Fox}. Hence,
\begin{align}\label{eq:residuos}
\frac{1}{2\pi \mathrm{i}}\!\int_{\mathcal{C}}^{}\!\Theta(s)z^{\!-\!s}ds\!=\!\!\sum_{h=0}^{\infty}\!\lim_{s \to \pm \chi(h)} \left (s\pm\chi(h) \right )\Theta(s)z^{-s},
\end{align} 
where $\chi(h)$ is a specific pole of $\Theta(s)$. Now, using~\eqref{eq:12} and~\eqref{eq:13}, $\mathrm{H}_2$ in~\eqref{eq:CDFRATIO} can be rewritten as
 \begin{equation}\label{eq:14}
\mathrm{H}_2=\frac{1}{2\pi \mathrm{i}} \int_{\mathcal{C}}^{   }\frac{\Gamma(s)\Gamma\left ( \mu_1-s \right )\Gamma(k\mu_1+\mu_2-k s)z^{-s}ds}{\Gamma(1+\mu_1-s)},
\end{equation}
where the suitable contour $\mathcal{C}$ separates all the poles of $\Gamma(s)$ to the left from those of $\Gamma\left ( \mu_1-s \right )$ and $\Gamma(k\mu_1+\mu_2-k s)$ to the right. 
Then, we can evaluate~\eqref{eq:14} as the sum of residues, as follows
 \begin{align}\label{eq:15}
\mathrm{H}_2= S_1+S_2,
\end{align}
where we have split the analysis of the Fox H-function given in~\eqref{eq:14} into two sums of residues\footnote{It is worth mentioning that $S_1$ corresponds to the sum of residues with respect to the pole of $\Gamma(s)$. On the other hand, $S_2$  corresponds to the sum of residues regarding the poles of $\Gamma\left ( \mu_1-s \right )$ and $\Gamma(k\mu_1+\mu_2-k s)$.}, according to the following ranges of values of $k$: $\left (i\right)$ $\chi(h)=-h$, for $k\leq 1$; $\left (ii\right)$ $\chi(h)=\mu_1+h$ and $\chi(h)=\frac{k\mu_1+\mu_2+h}{k}$, for $k\geq 1$. Now,   
by using~\eqref{eq:residuos} and the condition for $k\leq 1$ into~\eqref{eq:14}, the term $S_1$ can be formulated as

\begin{align}\label{eq:S1Residuo}
S_1&=\sum_{h=0}^{\infty}\lim_{s \to -h } \frac{\Gamma(s)\Gamma\left ( \mu_1-s \right )\Gamma(k\mu_1+\mu_2-k s)}{\left ( s+h \right )^{-1}\Gamma(1+\mu_1-s) z^{s}}\nonumber \\
&=\sum_{h=0}^{\infty}\frac{z^{h}\Gamma\left (k(h+\mu_1)+\mu_2\right )}{\left (-1\right )^h  (h+\mu_1) \Gamma\left ( 1+h \right )}.
\end{align}

Likewise, by using~\eqref{eq:residuos} and the condition for $k\geq 1$ into~\eqref{eq:14}, the term $S_2$ can be expressed as.

\begin{align}\label{eq:S2Residuo}
S_2&=- \sum_{h=0}^{\infty}\lim_{s \to h+\mu_1} \frac{\Gamma(s)\Gamma\left ( \mu_1-s \right )\Gamma(k\mu_1+\mu_2-k s)}{\left ( s-h-\mu_1 \right )^{-1}\Gamma(1+\mu_1-s)z^{s}} \nonumber \\
&- \sum_{h=0}^{\infty}\lim_{s \to \frac{k\mu_1+\mu_2+h}{k} } \frac{\Gamma(s)\Gamma\left ( \mu_1-s \right )\Gamma(k\mu_1+\mu_2-k s)}{\left ( s-\frac{k\mu_1+\mu_2+h}{k} \right )^{-1}\Gamma(1+\mu_1-s)z^{s}} \nonumber \\
&= \sum_{h=0}^{\infty}\frac{z^{-h-\mu_1} \Gamma\left (h+\mu_1\right )\Gamma\left (-hk+\mu_2\right )}{\left (-1\right )^{h-2} \Gamma \left ( 1-h \right )h!}\nonumber \\
&+ \sum_{h=0}^{\infty}\frac{z^{-\frac{h}{k}-\mu_1-\frac{\mu_2}{k}}\Gamma\left (\frac{-h-\mu_2}{k}\right )\Gamma\left (\frac{h+k\mu_1+\mu_2}{k}  \right )}{\left (-1\right )^{h-2} \Gamma \left ( \frac{-h+k-\mu_2}{k} \right ) kh!}.
\end{align}

By following a similar procedure as in the solution for $\mathrm{H}_2$, the series representation for $\mathrm{H}_1$ and $\mathrm{H}_3$ can be obtained as in~\eqref{eq:FoxbyResidues1} and~\eqref{eq:FoxbyResidues3}, respectively. 

\section{ Proof of Proposition~\ref{prop:moments}}
\label{ap:momentinv}
Let $Y_i$ be the inverse of $R_i$, with PDF given by~\cite[Eq.~(4)]{inversePDF}
\begin{equation}\label{eq:inversePDF}
f_{Y_i}(y)=\frac{\alpha_i \hat{r_i}^{\mu_i \alpha_i}  y^{-1-\alpha_i\mu_i}}{\mu_i^{\mu_i} \Gamma (\mu_i)}\exp\left(-\frac{\hat{r_i}^{\alpha_i}}{\mu_i y^{\alpha_i}} \right ).
\end{equation}
From~\eqref{eq:inversePDF}, the $n$th moment
$\mathbb{E}\left [ Y_i^n \right ]$ can be expressed as 
\begin{equation}\label{eq:inversemoments}
\mathbb{E}\left [ Y_i^n \right ]= \hat{r_i}^{n} \frac{ \Gamma\left ( \mu_i-n/\alpha_i  \right )}{\mu_i^{n/\alpha_i}  \Gamma (\mu_i)},  \hspace{2mm} n>\mu_i \alpha_i.
\end{equation}
Next, substituting $\mathbb{E}\left [\Upsilon_1^{n}\right ]=\mathbb{E}\left [R_1^{2n}\right ]$ by~\eqref{eq:moments}, and $\mathbb{E}\left [Z^n\right ]=\mathbb{E}\left [Y_2^{2n}\right ]$ by~\eqref{eq:inversemoments} into 
$\mathbb{E}\left [X^n\right ]=\mathbb{E}\left [\Upsilon_1^n\right ] \mathbb{E}\left [Z^{n}\right ]$, we obtain~\eqref{eq:HigherMoments}, thus completing the proof.


\end{document}